\documentclass[12pt,english]{article}
\usepackage{lmodern}
\usepackage[T1]{fontenc}
\usepackage[latin9]{inputenc}
\usepackage{geometry}
\usepackage{color}
\usepackage{babel}
\usepackage{float}
\usepackage{amsmath}
\usepackage{amsthm}
\usepackage{amssymb}
\usepackage{graphicx}
\usepackage{setspace}
\usepackage[authoryear]{natbib}
\usepackage[unicode=true,
 bookmarks=true,bookmarksnumbered=false,bookmarksopen=false,
 breaklinks=false,pdfborder={0 0 0},pdfborderstyle={},backref=false,colorlinks=true]
 {hyperref}
\hypersetup{pdftitle={Dynamic Information Design with Diminishing Sensitivity Over News},
 pdfpagelayout=OneColumn, pdfnewwindow=true, pdfstartview=XYZ, plainpages=false, urlcolor=[rgb]{0.0430 ,0, 0.5}, linkcolor=[rgb]{0.0430 ,0, 0.5}, citecolor=[rgb]{0.0430 ,0, 0.5}, hypertexnames=false}

\makeatletter

\providecommand{\tabularnewline}{\\}

\theoremstyle{definition}
\newtheorem{defn}{\protect\definitionname}
\theoremstyle{plain}
\newtheorem{prop}{\protect\propositionname}
\theoremstyle{plain}
\newtheorem{cor}{\protect\corollaryname}
\theoremstyle{plain}
\newtheorem{fact}{\protect\factname}
\theoremstyle{plain}
\newtheorem{lem}{\protect\lemmaname}

\usepackage{chngcntr}
\setcitestyle{round}
\usepackage{mathtools}
\usepackage{breakcites}
\usepackage[all]{hypcap}
\usepackage{dcolumn}

\makeatother

\providecommand{\corollaryname}{Corollary}
\providecommand{\definitionname}{Definition}
\providecommand{\factname}{Fact}
\providecommand{\lemmaname}{Lemma}
\providecommand{\propositionname}{Proposition}

\begin{document}
\title{Dynamic Information Design with \\Diminishing Sensitivity Over News\thanks{We thank Krishna Dasaratha, David Dillenberger, Ben Enke, Drew Fudenberg,
Simone Galperti, Jerry Green, Faruk Gul, David Hagmann, Marina Halac,
Johannes H\"{o}rner, David Laibson, Shengwu Li, Jonathan Libgober,
Elliot Lipnowski, Erik Madsen, Pietro Ortoleva, Matthew Rabin, Gautam
Rao, Collin Raymond, Joel Sobel, Charlie Sprenger, Tomasz Strzalecki,
the MIT information design reading group, and our seminar participants
for insightful comments. Kevin He thanks the California Institute
of Technology for hospitality when some of the work on this paper
was completed. All remaining errors are our own.}}
\author{Jetlir Duraj\thanks{University of Pittsburgh. Email: \protect\href{mailto:jetlirduraj@gmail.com}{jetlirduraj@gmail.com}}
\and Kevin He\thanks{University of Pennsylvania. Email: \protect\href{mailto:hesichao@gmail.com}{hesichao@gmail.com}}}
\date{{\normalsize{}}%
\begin{tabular}{rl}
First version: & July 1, 2019\tabularnewline
This version: & January 13, 2023\tabularnewline
\end{tabular}\vspace{-3ex}}
\maketitle
\begin{abstract}
A Bayesian agent experiences gain-loss utility each period over changes
in belief about future consumption (\textquotedblleft news utility\textquotedblright ),
with diminishing sensitivity over the magnitude of news. Diminishing
sensitivity induces a preference over news skewness: gradual bad news,
one-shot good news is worse than one-shot resolution, which is in
turn worse than gradual good news, one-shot bad news. So, the agent's
preference between gradual information and one-shot resolution can
depend on his consumption ranking of different states. In a dynamic
cheap-talk framework where a benevolent sender communicates the state
over multiple periods, the babbling equilibrium is essentially unique
without loss aversion. More loss-averse agents may enjoy higher news
utility in equilibrium, contrary to the commitment case. We characterize
the family of gradual good news equilibria that exist with high enough
loss aversion, and find the sender conveys progressively larger pieces
of good news. We discuss applications to media competition and game
shows.

\bigskip{}


{\normalsize{}\thispagestyle{empty} \setcounter{page}{0}}{\normalsize\par}

\newpage{}
\end{abstract}
\global\long\def\equivalent{\Longleftrightarrow}%
 
\global\long\def\imply{\Longrightarrow}%
 
\global\long\def\composedwith{\cdot}%
 
\global\long\def\ra{\rightarrow}%
 
\global\long\def\indiff{\sim}%
 
\global\long\def\better{\succeq}%
 
\global\long\def\sbetter{\succ}%
 
\global\long\def\worse{\preceq}%
 
\global\long\def\sworse{\prec}%

\section{\label{sec:Introduction}Introduction}

People are sometimes willing to pay a cost to change how they receive
news over time, even when the information does not help them make
better decisions. Consider the following scenario:

Ann interviews for her dream job and is told that she will receive
the decision by email next week. Ann knows that if the firm decides
to reject her, she will receive the rejection email next Friday. But
if the firm decides to hire her, she could hear back on any of the
weekdays \textemdash{} in other words, no news is bad news. To avoid
experiencing multiple instances of disappointment over the week in
case she does not hear back for several days, Ann sets up an email
filter to automatically redirect any emails from the firm into a holding
tank, then releases all messages from the holding tank into her inbox
at 5PM next Friday.

In this scenario, Ann may be willing to exert costly effort to modify
her informational environment because she experiences \emph{diminishingly
sensitive} psychological reactions to good and bad news. She is elated
by good news and disappointed by bad news in every period, and multiple
congruent pieces of news carry a greater total emotional impact if
they are experienced separately in different periods than if the aggregated
lump-sum news arrives in a single period. This kind of psychological
consideration also influences how people convey news to others. When
CEOs announce earnings forecasts to shareholders and when organization
leaders update their teams about recent developments, they are surely
mindful of their information's emotional impact (in addition to its
possible instrumental value). Finally, the psychological effects of
news also play a prominent role in designing entertainment content
like game shows, where the audience experiences positive and negative
reactions over time to news and developments that have no bearing
on their personal decision-making.

In this paper, we study the implications of diminishingly sensitive
reactions to news for informational preference and dynamic communication.
A person's future consumption depends on an unknown state of the world.
In each period, he observes some information about the state and experiences
gain-loss utility over the \emph{change} in his belief about said
future consumption (``\emph{news utility}''). How does this person
prefer to learn about the state over time? If there is another agent
who knows the state and who wants to maximize the first person's expected
welfare, how will this informed agent communicate her information?

Of course, we are not the first to model news utility (see \citet{kHoszegi2009reference})
or to study psychological considerations in dynamic games (see the
survey \citet{battigalli2022belief}, for example). Our main innovation
is the focus on the implications of diminishing sensitivity \textemdash{}
a classical but surprisingly under-studied assumption. Diminishing
sensitivity in reference dependence traces back to \citet{kahneman1979prospect}'s
original formulation of prospect theory. Based on Weber's law and
experimental findings about human perception, these authors envisioned
a gain-loss utility based on deviations from a reference point, where
larger deviations carry smaller marginal effects. But almost all subsequent
work on reference-dependent preferences use two-part linear gain-loss
utility functions, so their results are driven by loss aversion but
not diminishing sensitivity.\footnote{\citet{kHoszegi2009reference}'s model of news utility allows for
diminishing sensitivity and they argue it is a realistic feature.
But their results either work with a special case without diminishing
sensitivity, or are in a setting where news utility with and without
diminishing sensitivity are behaviorally equivalent.} Four decades since \citet{kahneman1979prospect}'s publication, \citet{odonoghue2018reference}'s
review of the ensuing literature summarizes the situation:
\begin{quotation}
``Most applications of reference-dependent preferences focus entirely
on loss aversion, and ignore the possibility of diminishing sensitivity
{[}...{]} The literature still needs to develop a better sense of
when diminishing sensitivity is important.''
\end{quotation}
We show that diminishing sensitivity leads to novel and testable predictions
for information design. First, when the agent commits to an information
structure ex-ante, diminishing sensitivity generates a preference
over the direction of news \emph{skewness}. Any information structure
where good news arrives all at once but bad news arrives gradually
in small pieces \textemdash{} such as waiting for the job offer in
Ann's scenario \textemdash{} is strictly worse than resolving all
uncertainty in one period (``one-shot resolution''). On the other
hand, any information structure with the opposite skewness \textemdash{}
good news arrives gradually but bad news all at once \textemdash{}
is strictly better than one-shot resolution, provided loss aversion
is weak enough. We relate this result to recent experiments about
preference over the skewness of information in Section \ref{subsec:Experiments-on-Information}. 

As \citet{kHoszegi2009reference} point out, the two-part linear news-utility
model (without diminishing sensitivity) predicts that people prefer
one-shot resolution over any other dynamic information structure.
At the same time, some other theories (e.g., \citet*{ely2015suspense}'s
suspense and surprise utility) make the ``opposite'' prediction
that one-shot resolution is the worst possible information structure.
By contrast, the skewness preference induced by news utility with
diminishing sensitivity implies the same person can make different
choices between gradual information and one-shot resolution in different
situations \textemdash{} in particular, it depends on his consumption
ranking over the states.

For instance, in a world where two possible states ($A$ and $B$)
are associated with two different consumption prizes, imagine  state
$A$ realizes if and only if a sequence of intermediate events all
take place successfully over time. We show that when the agent prefers
the prize in state $A$, he will choose to observe the intermediate
events resolve in real-time. But when he prefers the prize in state
$B$, he will choose to only learn the final state. At the population
level, this result shows that an underlying diversity in consumption
preferences within a society can create a diversity in informational
preferences, and suggests a mechanism for media competition. The result
also rationalizes a ``sudden death'' format often found in game
shows, where the contestant must overcome every challenge in a sequence
to win the grand prize (as opposed to the grand prize being contingent
on beating at least one of several challenges.)

Our second main result is that when an informed benevolent sender
communicates the state to the receiver through cheap talk, the receiver's
diminishing sensitivity leads to credibility problems for the sender.
We show that if the receiver has diminishing sensitivity and low enough
loss aversion, the lack-of-commitment problem is so severe that every
equilibrium is payoff-equivalent to the babbling equilibrium. The
reason is that the sender strictly prefers to lie and say the state
is good even when it is bad. This temptation is driven by the receiver's
diminishing sensitivity: even though the sender is far-sighted and
knows false hope creates additional disappointment when the state
is revealed, diminishing sensitivity limits the incremental disutility
of this extra future disappointment. Diminishing sensitivity thus
drives a wedge between the commitment solution and the equilibrium
outcome, whereas the two coincide without it. We also show that high
enough loss aversion can restore the equilibrium credibility of good-news
messages by increasing the future disappointment cost of false hope
in the bad state. As a consequence, receivers with higher loss aversion
may enjoy higher equilibrium payoffs.

With enough loss aversion, there exist non-babbling equilibria featuring
gradual good news. We characterize the entire family of such equilibria
and study how quickly the receiver learns the state. For a class of
news-utility functions that include a tractable quadratic specification,
the sender always conveys progressively larger pieces of good news
over time, so the receiver's equilibrium belief grows at an increasing
rate in the good state. The idea is that in equilibrium, the sender
must be made indifferent between giving false hope and telling the
truth in the bad state, and diminishing sensitivity implies that sustaining
said indifference requires a greater amount of false hope when the
receiver's current belief is more optimistic. This conclusion also
puts a uniform bound on the number of periods of informative communication
across all time horizons and all equilibria.

The rest of the paper is organized as follows. Section \ref{sec:Model}
defines the timing of events and introduces a model of news utility
with diminishing sensitivity. Section \ref{sec:gradual_vs_one_shot}
studies how diminishing sensitivity leads to a preference over information
structures with different skewness, then applies this result to show
how an agent's choice between gradual information and one-shot resolution
depends on his consumption ranking of the states. Section \ref{sec:The-Credibility-Problem}
considers an environment where an informed benevolent sender communicates
the state to a receiver with news utility, and focuses on the credibility
problems in the resulting cheap-talk game. Section \ref{sec:Some-Related-Models}
discusses related literature and contrasts our results with the predictions
of other models of preference over non-instrumental information. Section
\ref{sec:Concluding-Discussion} concludes. 

\section{\label{sec:Model}Model}

\subsection{Timing of Events}

We consider a discrete-time model with periods $0,1,2,...,$$T$,
where $T\ge2$. There is a binary state space $\Theta=\{A,B\}.$ In
the final period $T,$ the agent receives a state-dependent consumption
prize $c_{\theta}$ and derives consumption utility $v(c_{\theta})\in\mathbb{R}$.
There is no consumption in other periods, and we assume that $v(c_{A})\ne v(c_{B}).$
We may normalize $v$ without loss so that the agent gets consumption
utility 1 in one state and 0 in the other.

The agent starts with a prior probability $0<\pi_{0}<1$ of the state
being $A$. In every period $t=1,...,T,$ the agent observes some
information and updates his belief about $\{\theta=A\}$ to the Bayesian
posterior $0\le\pi_{t}\le1.$ The information is non-instrumental
in that no actions taken in these interim periods affect the state
or the consumption utility in period $T$. In period $T,$ he exogenously
and perfectly learns the true state $\theta$ at the moment of consumption,
so we always have $\pi_{T}=1$ if $\theta=A$ and $\pi_{T}=0$ if
$\theta=B$.

\subsection{News Utility}

Although the agent only consumes in the final period, he experiences
news utility over consumption in every period. He has a gain-loss
utility function, $\mu:[-1,1]\to\mathbb{R}$, that maps changes in
expected final-period consumption utility into a felicity level. Let
$\rho_{t}=\sum_{\theta\in\Theta}\pi_{t}(\theta)v(c_{\theta})$ denote
this expectation based on the agent's belief in period $t$, and note
$\rho_{t}\in[0,1]$ based on our normalization of $v.$ At the end
of period $1\le t\le T,$ the agent experiences news utility $\mu(\rho_{t}-\rho_{t-1})$
\textemdash{} that is, he derives joy or pain based on the recent
belief update from $\pi_{t-1}$ to $\pi_{t}.$ Utility flow is undiscounted
and the agent has the same $\mu$ in all periods,\footnote{Our preference satisfies \citet{segal1990two}'s time neutrality axiom.
We abstract away from preferences for early or late resolution of
uncertainty.} so his total payoff is $\sum_{t=1}^{T}\mu(\rho_{t}-\rho_{t-1})+v(c).$

Throughout we assume $\mu$ is continuous, strictly increasing, twice
differentiable except possibly at 0, and $\mu(0)=0.$ We maintain
further assumptions on $\mu$ to reflect diminishing sensitivity and
loss aversion.
\begin{defn}
\label{def:DM_loss_averse}Say $\mu$ satisfies\emph{ diminishing
sensitivity} if $\mu^{''}(x)<0$ and $\mu^{''}(-x)>0$ for all $x>0.$
Say $\mu$ satisfies \emph{(weak) loss aversion} if $-\mu(-x)\ge\mu(x)$
for all $x>0.$ There is \emph{strict loss aversion }if $-\mu(-x)>\mu(x)$
for all $x>0.$
\end{defn}
For instance, the gain-loss function $\mu$ in \citet{tversky1992advances}
where $\mu(x)=x^{\alpha}$ for $x\ge0,$ $\mu(x)=-\lambda|x|^{\beta}$
for $x<0$ with $0<\alpha,\beta<1$ and $\lambda>1$ satisfies both
diminishing sensitivity and strict loss aversion.

This model of diminishing sensitivity over the magnitude of news shares
the same psychological motivation as \citet{kahneman1979prospect},
who base their theory of human responses to monetary gains and losses
on Weber's law and on psychology experiments about how people perceive
changes in physical attributes like temperature or brightness. This
framework of deriving utility from changes in beliefs has been previously
discussed in \citet{kHoszegi2009reference}, but they mostly focus
on another model that makes percentile-by-percentile comparisons between
old and new beliefs and without diminishing sensitivity. The model
we use allows us to characterize the implications of diminishing sensitivity
in the simplest setup with two states. 

\subsubsection{Quadratic News Utility}

We discuss another tractable functional form of $\mu$ that is rich
enough to exhibit both diminishing sensitivity and loss aversion.
The quadratic news-utility function $\mu:[-1,1]\to\mathbb{R}$ is
given by

\[
\mu(x)=\begin{cases}
\alpha_{p}x-\beta_{p}x^{2} & x\ge0\\
\alpha_{n}x+\beta_{n}x^{2} & x<0
\end{cases}
\]
with $\alpha_{p},\beta_{p},\alpha_{n},\beta_{n}>0.$ So we have 
\[
\mu^{'}(x)=\begin{cases}
\alpha_{p}-2\beta_{p}x & x>0\\
\alpha_{n}+2\beta_{n}x & x<0
\end{cases},\qquad\mu^{''}(x)=\begin{cases}
-2\beta_{p} & x>0\\
2\beta_{n} & x<0
\end{cases}.
\]
The parameters $\alpha_{p},\alpha_{n}$ control the extent of loss
aversion near 0, while $\beta_{p},\beta_{n}$ determine the amount
of curvature \textemdash{} i.e., the second derivative of $\mu$.
The maintained general assumptions on $\mu$ imply the following parametric
restrictions.
\begin{enumerate}
\item \emph{Monotonicity}: $\alpha_{p}>2\beta_{p}$ and $\alpha_{n}>2\beta_{n}$.
These inequalities hold if and only if $\mu$ is strictly increasing.
\item \emph{Loss aversion}: $\alpha_{n}-\alpha_{p}\ge(\beta_{n}-\beta_{p})z$
for all $z\in[0,1]$. This condition is equivalent to loss aversion
from Definition \ref{def:DM_loss_averse} for this class of news-utility
functions.
\end{enumerate}
A family of quadratic news-utility functions that satisfy these two
restrictions can be constructed by choosing any $\alpha>2\beta>0$
and $\lambda\ge1$, then set $\alpha_{p}=\alpha$, $\alpha_{n}=\lambda\alpha$,
$\beta_{p}=\beta,$ $\beta_{n}=\lambda\beta$. Figure \ref{fig:Examples-of-quadratic}
plots some of these news-utility functions for different values of
$\alpha,\beta,$ and $\lambda$.

\begin{figure}
\begin{centering}
\includegraphics[scale=0.4]{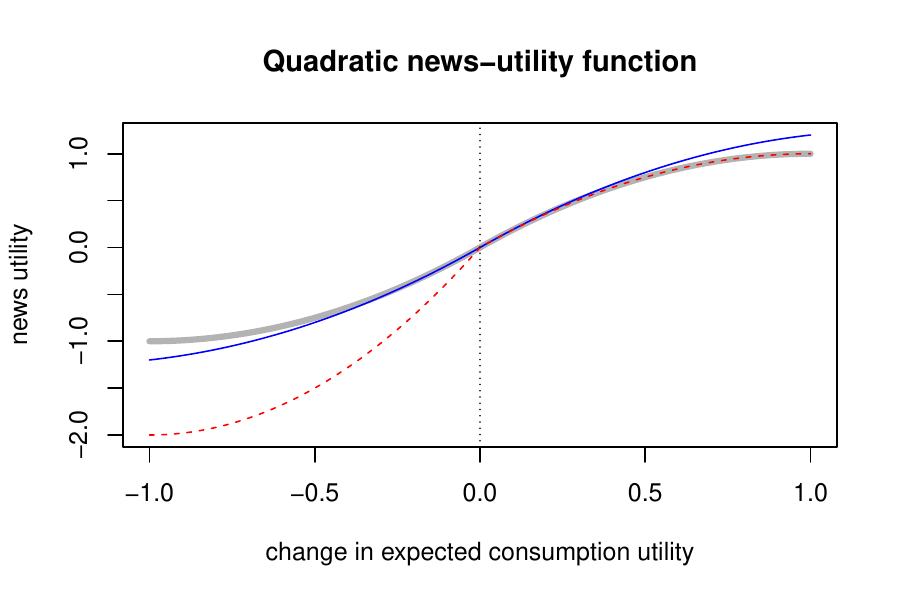}
\par\end{centering}
\caption{\label{fig:Examples-of-quadratic}Examples of quadratic news-utility
functions in the family $\alpha_{p}=\alpha$, $\alpha_{n}=\lambda\alpha$,
$\beta_{p}=\beta,$ $\beta_{n}=\lambda\beta$. Grey curve: $\alpha=2,$
$\beta=1,$ $\lambda=1$. Red curve: $\alpha=2$, $\beta=1$, $\lambda=2$.
Blue curve: $\alpha=2,$ $\beta=0.8$, $\lambda=1$.}
\end{figure}

\section{\label{sec:gradual_vs_one_shot}Diminishing Sensitivity and Preference
over News Skewness}

In this section, we show that news utility with diminishing sensitivity
makes novel predictions about preference over the skewness of information.
When there is no loss aversion, one-shot resolution of uncertainty
is neither the agent's most preferred way to get information nor the
least preferred. Instead, the agent strictly prefers one-shot resolution
over an information structure that delivers piecemeal bad news over
time, and strictly prefers an information structure with the opposite
skewness over one-shot resolution. By continuity, the same conclusions
hold when loss aversion is present but sufficiently weak.
\begin{defn}
\label{def:gradual_one_shot} An information structure features \emph{gradual
good news, one-shot bad news} if
\begin{itemize}
\item $\mathbb{P}[\rho_{t}\ge\rho_{t-1}\text{ for all }1\le t\le T\mid\rho_{T}=1]=1$
and
\item $\mathbb{P}[\rho_{t}<\rho_{t-1}\text{ for no more than one }1\le t\le T\mid\rho_{T}=0]=1$.
\end{itemize}
An information structure features \emph{gradual bad news, one-shot
good news} if
\begin{itemize}
\item $\mathbb{P}[\rho_{t}\le\rho_{t-1}\text{ for all }1\le t\le T\mid\rho_{T}=0]=1$
and
\item $\mathbb{P}[\rho_{t}>\rho_{t-1}\text{ for no more than one }1\le t\le T\mid\rho_{T}=1]=1$.
\end{itemize}
\end{defn}
The event $\rho_{T}=1$ corresponds to the ``good'' state being
realized and $\rho_{T}=0$ corresponds to the ``bad'' state being
realized, from the perspective of the agent's consumption utility.
In the gradual good news, one-shot bad news information structures,
the agent gets good news over time and gradually increases his expectation
of future consumption. When the state is bad, the agent gets all the
negative information at once \textemdash{} in the first period when
his expectation $\rho_{t}$ strictly decreases, he fully learns that
the state is bad. Conversely, ``gradual bad news, one-shot good news''
refers to the opposite kind of information structure.

An information structure features \emph{one-shot resolution} if $\mathbb{P}[\rho_{t}\ne\rho_{t-1}\text{ for at most one }1\le t\le T]=1.$
That is, almost surely the agent's belief only changes in one period
(including the final period when true state is perfectly revealed).
Note that one-shot resolution falls into both classes from Definition
\ref{def:gradual_one_shot}. We say that an information structure
features \emph{strictly gradual good news} if 
\[
\mathbb{P}[\rho_{t}>\rho_{t-1}\text{ and }\rho_{t^{'}}>\rho_{t^{'}-1}\text{ for two distinct }1\le t,t^{'}\le T\mid\rho_{T}=1]>0.
\]
That is, there is positive probability that the agent's expectation
strictly increases at least twice in periods $1$ through $T$. Similarly
define \emph{strictly gradual bad news}.

We now prove that whenever $\mu$ satisfies diminishing sensitivity
and (weak) loss aversion, information structures featuring strictly
gradual bad news, one-shot good news are \emph{strictly} worse than
one-shot resolution. The intuition is that information structures
in this class deliver small pieces of bad news but large clumps of
good news, which is the exact opposite of what the agent wants when
he experiences diminishing sensitivity to news.
\begin{prop}
\label{prop:opposite} Suppose $\mu$ satisfies diminishing sensitivity
and weak loss aversion. Any information structure featuring strictly
gradual bad news, one-shot good news provides strictly lower utility
than one-shot resolution in expectation, and almost surely weakly
lower utility ex-post.
\end{prop}
Proposition \ref{prop:opposite} identifies a class of information
structures that are worse than one-shot resolution for news utility
with diminishing sensitivity, distinguishing it from other models
of information preference where one-shot resolution is the worst possible
information structure. Utility models that make this other prediction
include suspense and surprise \citep*{ely2015suspense} and news utility
with a two-part linear, gain-loving (instead of loss-averse) value
function \citep*{chapman2018loss,goetteheterogeneity}.

Next, we show that if the agent has diminishing sensitivity but not
loss aversion, then information structures with strictly gradual good
news, one-shot bad news are \emph{strictly} better than one-shot resolution.
\begin{prop}
\label{prop:gradual_good_preferred}Suppose $\mu$ satisfies diminishing
sensitivity and it is symmetric around 0 with $-\mu(-x)=\mu(x)$ for
all $x\ge0$ (that is, it does not exhibit loss aversion). Any information
structure featuring strictly gradual good news, one-shot bad news
provides strictly higher utility than one-shot resolution in expectation,
and almost surely weakly higher utility ex-post.
\end{prop}
In \citet{kHoszegi2009reference}'s model of news utility without
diminishing sensitivity, one-shot resolution is optimal among all
information structures.\footnote{\citet{kHoszegi2009reference} showed this for their percentile-based
model of news utility with binary states, while \citet{dillenberger2018additive}
proved the same also holds for arbitrarily many states.} By contrast, Proposition \ref{prop:gradual_good_preferred} can be
combined with continuity to show that for news-utility functions with
diminishing sensitivity and a small enough amount of loss aversion,
there are information structures that are strictly better than one-shot
resolution. To make this precise, consider the parametric class of
\emph{$\lambda$-scaled news-utility functions}. We fix some $\tilde{\mu}_{pos}:[0,1]\to\mathbb{R}_{+}$,
strictly increasing and strictly concave with $\tilde{\mu}_{pos}(0)=0$,
and consider the family of news-utility functions given by $\mu_{\lambda}(x)=\tilde{\mu}_{pos}(x)$,
$\mu_{\lambda}(-x)=-\lambda\tilde{\mu}_{pos}(x)$ for $x>0$ as we
vary the loss aversion parameter $\lambda\ge1.$
\begin{cor}
\label{cor:gradual_good_loss_averse}Consider a class of $\lambda$-scaled
news-utility functions $(\mu_{\lambda})_{\lambda\ge1}$ and any information
structure featuring strictly gradual good news, one-shot bad news.
There exists some $\bar{\lambda}>1$ so that for any $1\le\lambda\le\bar{\lambda}$,
this information structure gives strictly higher utility than one-shot
resolution in expectation.
\end{cor}
In summary, provided loss aversion is low enough, diminishing sensitivity
induces the following preference ranking: gradual good news, one-shot
bad news is better than one-shot resolution, which is in turn better
than gradual bad news, one-shot good news. Section \ref{subsec:Experiments-on-Information}
discusses related experimental literature about preference over news
skewness, and Appendix \ref{sec:Optimal-Information-Structure} contains
additional results about preference over information structures.

\subsection{Consumption Preference and Information Preference}

An an application, we consider an environment where a sequence of
signal realizations gradually determine the binary state. We show
that agents with opposite consumption preferences over the two states
can exhibit opposite preferences between observing the signals as
they arrive or only learning the final state, because the same gradual
information translates into two different kinds of skewness for these
agents.

In each period $t=1,2,...,T$, a binary signal $X_{t}$ realizes,
where $\mathbb{P}[X_{t}=1]=q_{t}$ with $0<q_{t}<1$. Each $X_{t}$
is independent of the other ones. The signals determine the state.
If $X_{t}=1$ for all $t$, then the state is $A$. Otherwise, when
$X_{t}=0$ for at least one $t$, the state is $B$.\footnote{Equivalently, we can think of state $A$ having probability $\Pi_{t=1}^{T}q_{t}$
and state $B$ having the complementary probability. Conditional on
$\theta=A$, we always have $X_{1}=....=X_{T}=1.$ Conditional on
$\theta=B$, for a sequence of signal realizations $(x_{1},x_{2},...,x_{T})\in\{0,1\}^{T}$,
we have $\mathbb{P}[(X_{1},...,X_{T})=(x_{1},...,x_{T})\mid\theta=B]=\frac{\Pi_{t=1}^{T}q_{t}^{x_{t}}\cdot(1-q_{t})^{1-x_{t}}}{1-\Pi_{t=1}^{T}q_{t}}$
if at least one $x_{i}$ is 0, otherwise $\mathbb{P}[(X_{1},...,X_{T})=(x_{1},...,x_{T})\mid\theta=B]=0.$
This is an equivalent description of the joint distribution between
the state and the signals $(X_{t})_{t=1}^{T},$ but it is more natural
to think of the signals determining the state over time in the applications
we discuss below.} At time 0, the agent chooses between observing the realizations of
the signals $(X_{t})_{t=1}^{T}$ in real time (\emph{gradual information}),
or only learning the state of the world at the end of period $T$
(\emph{one-shot resolution}).

For a concrete example, imagine a televised debate between two political
candidates $A$ and $B$ where $A$ loses as soon as she makes a ``gaffe''
during the debate.\footnote{\citet{augenblick2018belief} use a similar example of political gaffes
to illustrate Bayesian belief movements.} If $A$ does not make any gaffes, then $A$ wins. In this example,
$\{X_{t}=1\}$ corresponds to the event that candidate $A$ does not
make a gaffe during the $t$-th minute of the debate. States $A$
and $B$ correspond to candidates $A$ and $B$ winning the debate.
An individual chooses between watching the debate live (i.e., observing
the stochastic process $(X_{t})$ in real time) or only reading the
outcome of the debate the following morning (i.e., one-shot resolution
 about the state).

The individual could be someone who benefits from candidate $A$ winning
the debate (that is, $v(c_{A})=1,$ $v(c_{B})=0$), or someone who
benefits from candidate $B$ winning the debate (that is, $v(c_{A})=0,$
$v(c_{B})=1$). For the first type of agent, the debate provides gradual
good news, one-shot bad news. For the second type of agent, the debate
provides gradual bad news, one-shot good news. It follows from Proposition
\ref{prop:opposite} and Corollary \ref{cor:gradual_good_loss_averse}
that these two types can make different choices about whether to watch
the debate.
\begin{cor}
\label{cor:diversity} Consider a class of $\lambda$-scaled news-utility
functions $(\mu_{\lambda})_{\lambda\ge1}$. For any $\lambda\ge1$,
the agent chooses one-shot resolution over gradual information when
$v(c_{A})=0,$ $v(c_{B})=1$. There exists some $\bar{\lambda}>1$
so that for any $1\le\lambda\le\bar{\lambda}$, the agent chooses
gradual information over one-shot resolution when $v(c_{A})=1,$ $v(c_{B})=0$.
\end{cor}
At the population level, Corollary \ref{cor:diversity} shows that
society can exhibit an\emph{ endogenous diversity} of information
preferences, driven by an underlying diversity of consumption preferences.
Individuals with the same news-utility function $\mu$ can nevertheless
choose to learn about the state of the world in two different ways,
if they have opposite rankings of the states in terms of their consumption
levels. So heterogeneous consumption preferences generate heterogeneous
information preferences.\footnote{\citet{kim2021temporal} document this heterogeneity in information
choice, finding that people pay less attention to political news when
their own political party is performing poorly. This shows that individuals
choose their information exposure as in our model, and also provides
evidence that people avoid getting piecemeal bad news and savor piecemeal
good news.}

This observation also suggests a possible mechanism for media competition:
if the realization of some state $A$ depends on a series of smaller
events, then some news sources may cover these small events in detail
as they happen, while other sources may choose to only report the
final outcome. If there is a heterogeneity of tastes over states in
the society, then viewers will sort between these two kinds of news
sources based on how they rank states $A$ and $B$ in terms of consumption.

At the individual level, Corollary \ref{cor:diversity} shows that
the same person may choose gradual information in one situation but
one-shot resolution in another, even if his news-utility function
remains stable. For example, if political candidate $X$ wins any
debate when and only when she does not make a gaffe, an agent may
choose to watch a debate between candidates $X$ and $Y$ but refuse
to watch a debate between candidates $X$ and $Z$, because he prefers
$X$ over $Y$ but $Z$ over $X.$

By contrast, many related theories about behavioral information preference
tend to predict that the agent either always prefers one-shot resolution
in all situations, or always prefers every other information structure
to one-shot resolution in all situations.

\begin{prop}
\label{prop:no_diversity}The following models predict that the agent
will not change his choice between gradual information and one-shot
resolution when the sign of $v(c_{A})-v(c_{B})$ changes.
\begin{enumerate}
\item News utility with a two-part linear $\mu$, where $\mu(x)=x$ for
$x\ge0$ and $\mu(x)=\lambda x$ for $x<0$, with any $\lambda\ge0$.
\item Anticipatory utility where the agent gets $u(\rho_{t})$ in period
$t$, with $u$ an increasing, weakly concave function.
\item \citet*{ely2015suspense}'s ``suspense and surprise'' utility.
\end{enumerate}
\end{prop}

\subsection{An Application to Game Shows}

A final consequence of Corollary \ref{cor:diversity} concerns the
design of game shows. Consider a game show featuring a single contestant
who will win either \$100,000 or nothing depending on her performance
across five rounds.\footnote{This can be thought of as a stylized payout structure for game shows
like \emph{American Ninja Warrior }and \emph{Who Wants to Be a Millionaire}.} The audience, empathizing with the contestant, derives news utility
$\mu(\pi_{t}-\pi_{t-1})$ at the end of round $t,$ where $\pi_{t}$
is the contestant's probability of winning the prize based on the
first $t$ rounds. One possible format (``sudden death'') features
five easy rounds each with $w=0.5^{1/5}\approx87\%$ winning probability,
where the contestant wins \$100,000 if she wins all five rounds. Another
possible format (``rep\^{e}chage'') involves five hard rounds each
with $1-w$ winning probability, but the contestant wins \$100,000
as soon as she wins any round. Both formats lead to the same distribution
over final outcomes and generate the same amount of suspense and surprise
utilities \`{a} la \citet*{ely2015suspense}. Corollary \ref{cor:diversity}
shows the first format induces more news utility than one-shot resolution
(which could correspond to not watching the game show and simply looking
up the contestant's outcome later) for audience members who are not
too loss averse, while the second format is worse than one-shot resolution
for all audience members. Consistent with this prediction, the vast
majority of game shows resemble the first format more than the second
format.

\section{\label{sec:The-Credibility-Problem}Diminishing Sensitivity and the
Credibility Problem}

So far, we have assumed the agent commits to an information structure
ex-ante. In many economic settings, it is instead an informed individual
who communicates the state to the agent over time. Such communication
often takes the form of unverifiable cheap-talk messages, especially
if the speaker wishes to convey inconclusive news about the state.

We consider a cheap-talk game between a receiver who experiences news
utility with diminishing sensitivity, and a benevolent sender who
knows the state and wishes to maximize the receiver's welfare. At
first glance, one may think that the sender can simply implement the
receiver's favorite information structure in the equilibrium of the
game, given that the two parties have aligned incentives. While this
is true with two-part linear news utility, we show that the receiver's
diminishing sensitivity leads to a credibility problem for the sender.

\subsection{Cheap Talk with an Informed and Benevolent Sender}

Let a finite set of cheap-talk messages $M$ with $|M|\ge2$ be fixed.
The sender learns the true state of the world $\theta\in\{A,B\}$
in period $t=0$. The receiver's consumption preference is common
knowledge, which we normalized to be $v(c_{A})=1,$ $v(c_{B})=0.$
In every period $t=1,2,...,T-1$, the sender conveys a message $m\in M$
to the receiver. The sender's communication strategy in period $t$
is given by a mixture over messages $\sigma_{t}(\cdot\mid h^{t-1},\theta)\in\Delta(M)$
that can depend on the history $h^{t-1}$ of messages so far and the
true state $\theta$. The sender cannot commit to how she will communicate
with the receiver in different states of the world.

The sender is benevolent and wants to maximize the receiver's welfare.
At the end of period $T$, if the receiver has experienced the belief
path $(\pi_{t})_{t=0}^{T}$, then the sender's total payoff in the
game is $\sum_{t=1}^{T}\mu\left(\pi_{t}-\pi_{t-1}\right)$ (we may
ignore the physical consumption term since neither party can affect
it). The state of the world determines the final belief $\pi_{T}$
and thus affects news utility in the final period, so the sender expects
different payoffs from sending the same sequence of messages in different
states.\footnote{In particular, this is not a cheap-talk game with state-independent
sender payoffs, as in \citet{lipnowski2020}.}

We analyze perfect-Bayesian equilibria of the cheap talk game, under
some off-path belief refinements.
\begin{defn}
\label{defn:pbe}A\emph{ perfect-Bayesian equilibrium} consists of
sender's strategy $\sigma^{*}=(\sigma_{t}^{*})_{t=1}^{T-1}$ together
with receiver's beliefs $p^{*}:\cup_{t=0}^{T-1}H^{t}\to[0,1]$, where:
\begin{itemize}
\item For every $1\le t\le T-1,$ $h^{t-1}\in H^{t-1}$ and $\theta\in\{A,B\},$
$\sigma^{*}$ maximizes the receiver's total expected news utility
in periods $t,...,T-1,T$ conditional on having reached the public
history $h^{t-1}$ in state $\theta$ at the start of period $t$.
\item $p^{*}$ is derived by applying the Bayes' rule to $\sigma^{*}$ whenever
possible.
\end{itemize}
We make two belief-refinement restrictions:
\begin{itemize}
\item If $t\le T-1$, $h^{t}$ is a continuation history of $h^{\underline{t}}$,
and $p^{*}(h^{\underline{t}})\in\{0,1\},$ then $p^{*}(h^{t})=p^{*}(h^{\underline{t}}).$
\item The receiver's belief $\pi_{T}$ in period $T$ when state is $\theta$
puts probability 1 on $\theta$, regardless of the preceding history
$h^{T-1}\in H^{T-1}.$
\end{itemize}
\end{defn}
We will abbreviate a perfect-Bayesian equilibrium satisfying our off-path
belief refinements as an ``equilibrium.'' Our definition requires
that once the receiver updates his belief to 0 or 1, this belief stays
constant through the end of period $T-1$. In other words, the support
of his belief is non-expanding through the penultimate period.\footnote{This standard refinement was first used in \citet{grossman1986sequential}.
It rules out pathological off-path belief updates if the sender deviates
and sends a message perfectly indicative of one state following a
history where the receiver is fully convinced of the other state.} In period $T,$ the receiver updates his belief to reflect full confidence
in the true state of the world, regardless of his (possibly dogmatically
wrong) belief at the end of period $T-1$.

\emph{Babbling equilibria} always exists for any news-utility function
$\mu,$ message space $M,$ time horizon $T,$ and prior $\pi_{0}$.
In a babbling equilibrium, the sender mixes over messages in a state-independent
way, and the receiver keeps his prior belief $\pi_{0}$ after every
history up until period $T$. A babbling equilibrium implements one-shot
resolution for the receiver, as his belief stays constant and fully
resolves in the final period.

\subsection{The Credibility Problem and Babbling}

Are there equilibria where the sender gets a higher expected payoff
than the \emph{babbling payoff} of $\pi_{0}\mu(1-\pi_{0})+(1-\pi_{0})\mu(-\pi_{0})$?
By Proposition \ref{prop:gradual_good_preferred}, for a receiver
who has diminishing sensitivity but not loss aversion, there exist
a class of information structures that strictly improve on one-shot
resolution. But the next result proves none of these information structures
can be implemented in equilibrium.
\begin{prop}
\label{prop:unique_babbling}Suppose $\mu$ is symmetric around 0
and $\mu^{''}(x)<0$ for all $x>0$. For any $M,T,\pi_{0},$ the sender's
payoff in every equilibrium is equal to the babbling payoff.
\end{prop}
To understand why, consider the final period of communication $t=T-1.$
Suppose the state is bad and the sender must decide between revealing
the truth to decrease the receiver's belief from $\pi$ to 0, or sending
a positive message that increases the receiver's belief by $z>0$.
Such false hope in period $T-1$ gives positive news utility today
at the cost of increasing disappointment in the final period. But
diminishing sensitivity implies the marginal utility of positive news
today is larger than the marginal disutility of the incremental future
disappointment, $\mu(z)>\mu(-\pi)-\mu(-(\pi+z))$. This shows that
in equilibrium, the sender cannot communicate good news in either
state, otherwise she will be tempted to mimic the good-news messages
when the state is bad, destroying the credibility of these messages.
So the sender must babble in period $T-1,$ so we could treat period
$T-2$ as the last period of communication, and apply the same arguments
by backward induction.

In summary, diminishing sensitivity leads to a credibility problem
that prevents any informative communication, even though the players
share the same payoff function. In a cheap-talk setting with instrumental
information and anticipatory utility, \citet{kHoszegi2006emotional}
shows that a benevolent sender also distorts equilibrium communication
relative to the commitment benchmark. The breakdown in communication
is more complete in our setting, for the players get the same payoffs
as when communication is impossible.

The intuition we gave for the uniqueness of babbling up to payoffs
assumes the receiver is not loss averse \textemdash{} that is, $\mu$
is symmetric around 0. Babbling remains unique with a small amount
of loss aversion, but a high enough level of loss aversion can restore
the sender's credibility and enable non-babbling equilibria. (In the
next section, we will construct a family of such non-babbling equilibria.)

\begin{figure}
\begin{centering}
\includegraphics[scale=0.4]{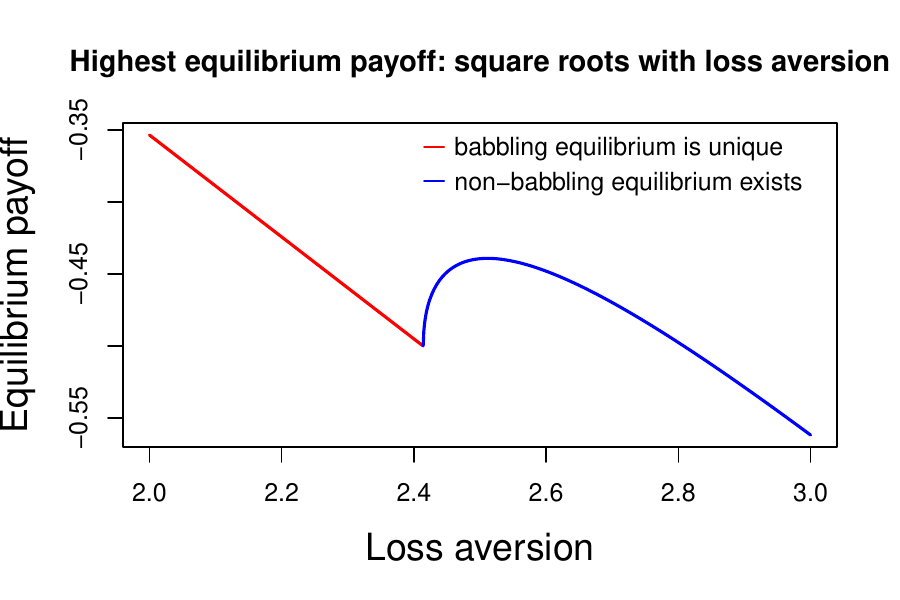}
\par\end{centering}
\caption{\label{fig:Highest-equilibrium-payoff} The babbling equilibrium is
unique up to payoffs for low values of $\lambda$, but there exists
an equilibrium with gradual good news for $\lambda\ge2.414$. Due
to the role of loss aversion in sustaining credible partial news,
a receiver with higher loss aversion may experience higher or lower
expected news utility in equilibrium than a receiver with lower loss
aversion.}
\end{figure}
 To illustrate, suppose $\mu(x)=\sqrt{x}$ for $x\ge0,$ $\mu(x)=-\lambda\sqrt{-x}$
for $x<0,$ $T=2$, and $\pi_{0}=\frac{1}{2}$. Figure \ref{fig:Highest-equilibrium-payoff}
plots the highest equilibrium payoff for different values of $\lambda.$
Receivers with higher $\lambda$ may enjoy higher equilibrium payoffs.
The reason for this non-monotonicity is that for low values of $\lambda$,
the babbling equilibrium is unique and increasing $\lambda$ decreases
expected news utility linearly. When the new, non-babbling equilibrium
emerges for large enough $\lambda$, the sender's behavior in the
new equilibrium depends on $\lambda$. Higher loss aversion carries
two countervailing effects: first, a \emph{non-strategic effect} of
hurting welfare when $\theta=B$, as the receiver must eventually
hear the bad news; second, an \emph{equilibrium effect} of changing
the relative amounts of good news in different periods conditional
on $\theta=A$. Receivers with an intermediate amount of loss aversion
enjoy higher expected news utility than receivers with low loss aversion,
as the equilibrium effect leads to better ``consumption smoothing''
of good news across time. But, the non-strategic effect eventually
dominates and receivers with high loss aversion experience worse payoffs
than receivers with low loss aversion.

\subsection{Deterministic Gradual Good News Equilibria}

When the receiver's loss aversion is high enough, there can exist
non-babbling equilibria in the cheap-talk game. We now analyze a family
of such non-babbling equilibria, where the receiver's belief monotonically
increases over time conditional on the good state. These equilibria
show that the gradual good news, one-shot bad news information structures
discussed in Section \ref{sec:gradual_vs_one_shot} can be sustained
without commitment.

An equilibrium $(M,\sigma^{*},p^{*})$ features \emph{deterministic}\footnote{This class of equilibria is slightly more restrictive than the gradual
good news, one-shot bad news information structures from Definition
\ref{def:gradual_one_shot}, because the sender may not randomize
between several increasing paths of beliefs in the good state.}\emph{gradual good news} (GGN equilibrium) if there exist a sequence
of constants $p_{0}\le p_{1}\le...\le p_{T-1}\le p_{T}$ with $p_{0}=\pi_{0}$,
$p_{T}=1$, and the receiver always has belief $p_{t}$ in period
$t$ when the state is good. By Bayesian beliefs, in the bad state
of any GGN equilibrium the sender must induce a belief of either $0$
or $p_{t}$ in period $t$, as any message not inducing belief $p_{t}$
is a conclusive signal of the bad state.

The class of GGN equilibria is non-empty, for it contains the babbling
equilibrium where $\pi_{0}=p_{0}=p_{1}=...=p_{T-1}<p_{T}=1$. The
number of \emph{intermediate beliefs} in a GGN equilibrium is the
number of distinct beliefs in the open interval $(\pi_{0},1)$ along
the sequence $p_{0},p_{1},...,p_{T-1}$. The babbling equilibrium
has zero intermediate beliefs.

The next proposition characterizes the set of all GGN equilibria with
at least one intermediate belief.
\begin{prop}
\label{prop:gradual_family}Let $P^{*}(\pi)\subseteq(\pi,1]$ be those
beliefs $x>\pi$ satisfying $\mu(x-\pi)+\mu(-x)=\mu(-\pi).$ Suppose
$\mu$ exhibits diminishing sensitivity and loss aversion. For $1\le J\le T-1,$
there exists a gradual good news equilibrium with the $J$ intermediate
beliefs $q^{(1)}<...<q^{(J)}$ if and only if $q^{(j)}\in P^{*}(q^{(j-1)})$
for every $j=1,...,J$, where $q^{(0)}:=\pi_{0}$.
\end{prop}
To interpret, $P^{*}(\pi)$ contains the set of beliefs $x>\pi$ such
that the sender is indifferent between inducing the two belief paths
$\pi\to x\to0$ and $\pi\to0.$ When $\mu$ is symmetric, this indifference
condition is never satisfied, which is the source of the credibility
problem for good-news messages. The same indifference condition pins
down the relationship between successive intermediate beliefs in GGN
equilibria. This condition ensures that in the bad state, the sender
is willing to randomize between revealing the state and lying with
an inconclusive piece of good news that moves the receiver to the
next intermediate belief.

We illustrate this result with the quadratic news utility.
\begin{cor}
\label{cor:P_star_set_quadratic} 1) With quadratic news utility,
$P^{*}(\pi)=\left\{ \pi\cdot\frac{\beta_{p}+\beta_{n}}{\beta_{p}-\beta_{n}}-\frac{\alpha_{n}-\alpha_{p}}{\beta_{p}-\beta_{n}}\right\} \cap(\pi,1).$

2a) If $\beta_{n}>\beta_{p}$, there cannot exist any gradual good
news equilibrium with more than one intermediate belief.

2b) If $\beta_{n}<\beta_{p}$, there can exist gradual good news equilibria
with more than one intermediate belief. For a given set of parameters
of the quadratic news-utility function and prior $\pi_{0}$, there
exists a uniform bound on the number of intermediate beliefs that
can be sustained in equilibrium across all $T$.

3) In any GGN equilibrium with quadratic news utility, intermediate
beliefs in the good state grow at an increasing rate.
\end{cor}
For the case of quadratic news utility, this result provides a closed-form
characterization of the successive intermediate beliefs. It also shows
every GGN equilibrium involves progressively larger pieces of good
news in the good state, $q^{(j+1)}-q^{(j)}>q^{(j)}-q^{(j-1)}$. The
convex time-path of equilibrium beliefs is due to diminishing sensitivity.
If the sender is indifferent between providing $d$ amount of false
hope and truth-telling in the bad state when the receiver has prior
belief $\pi_{L}$, then she strictly prefers providing the same amount
of false hope over truth-telling at any more optimistic prior belief
$\pi_{H}>\pi_{L}$. The false hope generates the same positive news
utility in both cases, but an extra $d$ units of disappointment matters
less when added a baseline disappointment level of $\pi_{H}$ rather
than $\pi_{L}$, thanks to diminishing sensitivity.

Equilibrium beliefs in the good state grow at an increasing rate,
but must be bounded above by 1. So, there exists some uniform bound
$\bar{J}$ on the number of intermediate beliefs depending only on
the prior belief $\pi_{0}$ and parameters of the news-utility function.

As an illustration, consider the quadratic news utility with $\alpha_{p}=2$,
$\alpha_{n}=2.1$, $\beta_{p}=1$, and $\beta_{n}=0.2$. Starting
at the prior belief of $\pi_{0}=\frac{1}{3}$, Figure \ref{fig:GGN_example}
shows the longest possible sequence of intermediate beliefs in any
GGN equilibrium for arbitrarily large $T$. Since the $P^{*}$ sets
are either empty sets or singleton sets for the quadratic news utility,
Figure \ref{fig:GGN_example} also contains all the possible beliefs
in any state of any GGN equilibrium with these parameters.
\begin{flushleft}
\begin{figure}
\begin{centering}
\includegraphics[scale=0.35]{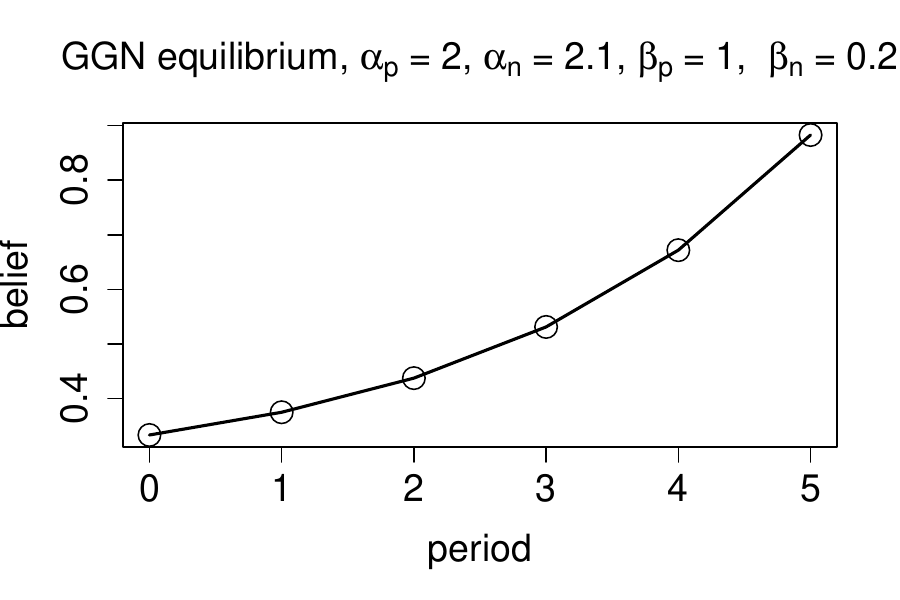}
\par\end{centering}
\caption{\label{fig:GGN_example}The longest possible sequence of GGN intermediate
beliefs starting with prior $\pi_{0}=\frac{1}{3}$. For quadratic
news utility, equilibrium GGN beliefs always increase at an increasing
rate in the good state.}
\end{figure}
Beyond the quadratic case, the intuition that diminishing sensitivity
should cause the receiver to have a convex time-path of equilibrium
beliefs holds more generally. The next result formalizes this relationship.
It shows that when diminishing sensitivity is combined with a pair
of regularity conditions, intermediate beliefs grow at an increasing
rate in any GGN equilibrium. These conditions are satisfied, for example,
by the square-roots news utility with loss aversion.
\par\end{flushleft}
\begin{prop}
\label{prop:ggn_increasing_rate}Suppose $\mu$ exhibits diminishing
sensitivity, $|P^{*}(\pi)|\le1$ and $\frac{\partial}{\partial\epsilon}\left[\mu(\epsilon)+\mu(-(\pi+\epsilon))\right]|_{\epsilon=0}>0$
for all $\pi\in(0,1)$. Then, in any GGN equilibrium with intermediate
beliefs $q^{(1)}<...<q^{(J)}$, we get $q^{(j)}-q^{(j-1)}<q^{(j+1)}-q^{(j)}$
for all $1\le j\le J-1$.
\end{prop}
The first regularity condition requires that the sender is indifferent
between the belief paths $\pi\to x\to0$ and $\pi\to0$ for at most
one $x>\pi.$ It is a technical assumption that lets us prove our
result, but we suspect the conclusion also holds under some relaxed
conditions. The second regularity condition says in the bad state,
the total news utility associated with an $\epsilon$ amount of false
hope is higher than truth-telling for small $\epsilon$.

\section{Related Literature and Predictions of Other Belief-Based Utility
Models \label{sec:Some-Related-Models}}

\subsection{\label{subsec:Experiments-on-Information}Experiments on Information
Preference}

A number of experimental papers have tested whether people prefer
one-shot resolution by asking subjects to choose how they wish to
learn about their prize for the experiment, with one-shot resolution
as a feasible information structure. The empirical results are mixed.
After accounting for preference over the timing of resolution,\footnote{Information structures that reveal the prize gradually will resolve
uncertainty earlier than a one-shot resolution structure that reveals
the prize at the end of the experiment, but later than a one-shot
resolution structure that reveals the prize immediately.} \citet{falk2016beliefs} and \citet*{bellemare2005myopic} find evidence
that subjects prefer one-shot resolution, while \citet*{nielsen2020jet,masatlioglu2017intrinsic,zimmermann2014clumped,budescu2001same}
find evidence against it. News utility with diminishing sensitivity
may explain these mixed results, as it predicts one-shot resolution
is neither the best nor the worst information structure, so it may
or may not be chosen depending on what other information structures
are feasible in a particular experiment. On the other hand, these
experimental results are harder to reconcile with theories that either
predict agents always choose one-shot resolution or predict agents
always avoid it.

Two experiments have examined people's preference over the skewness
of news, with mixed results. Tables 10 and 11 in \citet{nielsen2020jet}
report that subjects prefer negatively skewed news, as predicted by
news utility with diminishing sensitivity. But, \citet*{masatlioglu2017intrinsic}
find that agents prefer positively skewed news. In showing that a
classical assumption of reference dependence leads to a prediction
about preference over news skewness, we hope to stimulate further
empirical work on this topic. 

Finally, consistent with the mechanism discussed in Section \ref{sec:gradual_vs_one_shot},
\citet*{gul_experiment} find in an experiment that more people choose
gradual information over early one-shot resolution when the gradual
information features gradual good news rather than gradual bad news.

\subsection{Related Work on New Utility}

Since \citet{kHoszegi2009reference}, several other authors have analyzed
the implications of news utility in different settings: asset pricing
\citep{pagel2016expectations}, life-cycle consumption \citep{pagel2017expectations},
portfolio choice \citep{pagel2018news}, and mechanism design \citep{duraj2018mechanism}.
These papers focus on Bayesian agents with two-part linear gain-loss
utilities and do not study the role of diminishing sensitivity to
news.

Interpreting monetary gains and losses as news about future consumption,
experiments that show risk-seeking behavior when choosing between
loss lotteries and risk-averse behavior when choosing between gain
lotteries provide evidence for diminishing sensitivity over consumption
news (see e.g., \citet{rabin2009narrow}). In the same vein, papers
in the finance literature that use diminishing sensitivity over monetary
gains and losses to explain the disposition effect \citep*{shefrin1985disposition,kyle2006prospect,barberis2012realization,henderson2012prospect}
also provide indirect evidence for diminishing sensitivity over consumption
news.

\citet*{bowman1999loss} study a consumption-based reference-dependent
model with diminishing sensitivity. A critical difference is that
their reference points are based on past habits, not rational expectations.
We are not aware of existing work that focuses on how diminishing
sensitivity matters for information design with news utility.

\subsection{Predictions of Other Belief-Based Utility Models}

In general, papers on belief-based utility have highlighted two sources
of felicity: \emph{levels} of belief about future consumption utility
(``anticipatory utility,'' e.g., \citet{kHoszegi2006emotional,eliaz2006can,schweizer2018optimal})
and \emph{changes} in belief about future consumption utility (``news
utility'' and ``suspense and surprise'' \citep*{ely2015suspense}).
For the latter, some function of both the prior belief and the posterior
belief serves as the carrier of utility, while a given posterior belief
brings the same anticipatory utility for all priors \citep{eliaz2006can}.
The rich information preference under news utility with diminishing
sensitivity contrast against more stark predictions of the other commonly
used models of belief-based utility in the behavioral literature.

\subsubsection{News Utility without Diminishing Sensitivity}

The literature on reference-dependent preferences and news utility
has focused on two-part linear gain-loss utility functions, which
violate diminishing sensitivity. If $\mu$ is two-part linear with
loss aversion, then it follows from the martingale property of Bayesian
beliefs that one-shot resolution is weakly optimal for the agent among
all information structures. If there is strict loss aversion, then
one-shot resolution does strictly better than any information structure
that resolves uncertainty gradually. 

\subsubsection{Anticipatory Utility}

In our setup, an agent who experiences anticipatory utility gets $A\left(\sum_{\theta\in\Theta}\pi_{t}(\theta)\cdot v(c_{\theta})\right)$
if he ends period $t$ with posterior belief $\pi_{t}\in\Delta(\Theta),$
where $A:\mathbb{R}\to\mathbb{R}$ is a strictly increasing anticipatory-utility
function. When $A$ is the identity function (as in \citet{kHoszegi2006emotional}),
the solution to the optimization problem would be unchanged if we
modified our model and let the agent experience both anticipatory
utility and news utility. This is because by the martingale property,
the agent's ex-ante expected anticipatory utility in a given period
is the same across all information structures. So, the ranking of
information structures entirely depends on the news utility they generate.

For a general $A$, if the agent only experiences anticipatory utility,
not news utility, then there exists an optimal information structure
that only releases information in $t=1$, followed by uninformative
signals in all subsequent periods (see Online Appendix \ref{subsec:Optimal-Information-Anticipatory}).
By contrast, this kind of one-shot resolution is not optimal when
the agent has diminishing sensitivity and weak enough loss aversion. 

\subsubsection{Suspense and Surprise}

\citet*{ely2015suspense} study dynamic information design with a
Bayesian receiver who derives utility from suspense or surprise. They
propose and study an original utility function over belief paths where
larger belief movements always bring greater felicity. By contrast,
because our states are associated with different consumption consequences,
changes in beliefs may increase or decrease the receiver's utility
depending on whether the news is good or bad. While one-shot resolution
is suboptimal in both \citet*{ely2015suspense}'s problem and our
problem (under some conditions), other results differ. For example,
information structures featuring gradual bad news, one-shot good news
are worse than one-shot resolution in our problem, while one-shot
resolution is the worst possible information structure in \citet*{ely2015suspense}'s
problem.

\citet*{ely2015suspense} also discuss state-dependent versions of
suspense and surprise utilities, but this extension does not embed
our model. Suppose there are two states, $\Theta=\{G,B\},$ and the
agent has the suspense objective $\sum_{t=0}^{T-1}u\left(\mathbb{E}_{t}(\sum_{\theta}\alpha_{\theta}\cdot(\pi_{t+1}(\theta)-\pi_{t}(\theta))^{2}\right)$
or the surprise objective $\sum_{t=1}^{T}u\left(\sum_{\theta}\alpha_{\theta}\cdot(\pi_{t}(\theta)-\pi_{t-1}(\theta))^{2}\right)$,
where $\alpha_{G},\alpha_{B}>0$ are state-dependent scaling weights.
We must have $\pi_{t+1}(G)-\pi_{t}(G)=-(\pi_{t+1}(B)-\pi_{t}(B))$,
so pathwise $(\pi_{t+1}(G)-\pi_{t}(G))^{2}=(\pi_{t+1}(B)-\pi_{t}(B))^{2}$.
This shows that the new objectives obtained by applying two possibly
different scaling weights $\alpha_{G}\ne\alpha_{B}$ to states $G$
and $B$ are identical to the ones that would be obtained by applying
the \emph{same} scaling weight $\alpha=\frac{\alpha_{G}+\alpha_{B}}{2}$
to both states. Due to this symmetry in preference, the optimal information
structure for entertaining an agent with state-dependent suspense
or surprise utility treats the two states symmetrically, in contrast
to a central prediction of diminishing sensitivity in our model.

\subsubsection{Designing Beliefs through Non-Informational Channels}

\citet{brunnermeier2005optimal} and \citet{macera2014dynamic} study
the optimal design of beliefs for agents with belief-based utilities
that differ from the news-utility setup we consider. Another important
distinction is that we focus on the design of \emph{information}:
changes in the agent's belief derive from Bayesian updating an exogenous
prior, using the information conveyed by an information structure
or by a sender. \citet{macera2014dynamic} considers a non-Bayesian
agent who freely chooses a path of beliefs, while knowing the actual
state of the world. \citet{brunnermeier2005optimal} study the ``opposite''
problem to ours, where the agent freely chooses a prior belief (over
the sequence of state realizations) at the start of the game, then
updates belief about future states through an exogenously given information
structure.

\subsection{Related Decision-Theoretic Work on Information Preference}

Several paper in decision theory have studied models of preference
over dynamic information structures. \citet{dillenberger2010preferences}
shows that preference for one-shot resolution of uncertainty is equivalent
to a weakened version of independence, provided the preference satisfies
recursivity. This result does not apply here because our mean-based
model of news utility violates recursivity \textemdash{} it can be
shown that a news-utility agent may strictly prefer a 0\% chance of
winning a prize over a 1\% chance of winning it, if he will gradually
learn about the outcome of the lottery and has high enough loss aversion
(see Online Appendix \ref{subsec:Preference-for-Dominated}). \citet{dillenberger2018additive}
axiomatize a general class of additive belief-based preferences in
the domain of two-stage lotteries, relaxing recursivity and the independence
axiom. In the case of $T=2,$ our news-utility model belongs to the
class they characterize. Under this specialization, our work may be
thought of as studying the information design problem, with and without
commitment, using some of \citet{dillenberger2018additive}'s additive
belief-based preferences. \citet*{gul2019random} axiomatize a class
of preferences over non-instrumental information called risk consumption
preferences, including a novel ``peak-trough'' utility specification.
In contrast, we study the implications diminishing sensitivity, a
classical assumption from the behavioral economics literature. Our
model is not a risk consumption preference (see Online Appendix \ref{subsec:Risk-Consumption-Preferences}).

\subsection{Related Work in Dynamic Information Design}

In a setting without behavioral preferences, \citet{li2018sequential}
and \citet{wu2018sequential} consider a group of senders with commitment
power, sequentially sending signals to persuade a single receiver.
The receiver takes an action after observing all signals. In their
settings, every equilibrium in their setting can be converted into
a payoff-equivalent ``one-step'' equilibrium where the first sender
sends the joint signal implied by the old equilibrium, while all subsequent
senders babble uninformatively. Dynamics matter more in our setting,
as different sequences of interim beliefs cause the agent to experience
different amounts of total news utility.

\citet{lipnowski2018disclosure} study a static model of information
design with a psychological receiver whose welfare depends directly
on posterior belief. They discuss an application to a mean-based news-utility
model \emph{without} diminishing sensitivity in their Appendix A,
finding that either one-shot resolution or no information is optimal.
We focus on the implications of diminishing sensitivity. Our work
also differs in that we study a dynamic problem and examine equilibria
without commitment.

\section{\label{sec:Concluding-Discussion}Conclusion}

In this work, we have studied how diminishingly sensitive gain-loss
utilities applied to changes in beliefs affect the agent's informational
preferences. If we think that diminishing sensitivity to the magnitude
of news is psychologically realistic in this domain, then the stark
predictions of the ubiquitous two-part linear models may be misleading.
In the presence of diminishing sensitivity, richer informational preferences
emerge.

An agent's consumption preference over the states can determine his
preference between an information structure that delivers news gradually
and another that results in one-shot resolution. In general, one-shot
resolution is neither the best way to get information nor the worst
way \textemdash{} skewness matters. One-shot resolution is strictly
better than information structures with strictly gradual bad news,
one-shot good news. But, it is strictly worse than information structures
with strictly gradual good news, one-shot bad news, provided loss
aversion is not too high.

For an informed sender who lacks commitment power, diminishing sensitivity
leads to novel credibility problems that inhibit any meaningful communication
when the receiver has no loss aversion. High enough loss aversion
can restore the equilibrium credibility of good-news messages, and
the receiver's equilibrium welfare may be non-monotonic in loss aversion.
We construct a family of non-babbling equilibria with gradual good
news when loss aversion is high enough, finding that the sender must
communicate increasingly larger pieces of good news over time in the
good state.

\bibliographystyle{ecta}
\bibliography{news_utility_info}

\appendix
\begin{center}
\textbf{\Large{}Appendix}{\Large\par}
\par\end{center}

\renewcommand{\thecor}{A.\arabic{cor}}
\renewcommand{\theprop}{A.\arabic{prop}}
\renewcommand{\thelem}{A.\arabic{lem}}
\renewcommand{\thefigure}{A.\arabic{figure}}
\renewcommand{\thetable}{A.\arabic{table}}
\setcounter{cor}{0}
\setcounter{lem}{0}
\setcounter{prop}{0}
\setcounter{figure}{0}
\setcounter{table}{0}

\section{\label{sec:Omitted-Proofs}Proofs of the Main Results}

This appendix contains the proofs of the results stated in the main
text. Proofs of auxiliary results stated in the appendix appear in
Online Appendix \ref{sec:Secondary_proofs}.

In the proofs, we will often use the following fact about news-utility
functions with diminishing sensitivity. We omit its simple proof.
\begin{fact}
Let $d_{1},d_{2}>0$ and suppose $\mu(0)=0$.
\begin{itemize}
\item (sub-additivity in gains) If $\mu^{''}(x)<0$ for all $x>0$, then
$\mu(d_{1}+d_{2})<\mu(d_{1})+\mu(d_{2}).$
\item (super-additivity in losses) If $\mu^{''}(x)>0$ for all $x<0,$ then
$\mu(-d_{1}-d_{2})>\mu(-d_{1})+\mu(-d_{2})$
\end{itemize}
\end{fact}

\subsection{Proof of Proposition \ref{prop:opposite}}
\begin{proof}
In the less preferred state, the agent gets $\mu(-\rho_{0})$ with
one-shot resolution , but $\sum_{t=1}^{T}\mu(\rho_{t}-\rho_{t-1})$
with gradual bad news, one-shot good news. For each $t,$ $\rho_{t}-\rho_{t-1}\le0,$
and furthermore $\sum_{t=1}^{T}\rho_{t}-\rho_{t-1}=-\rho_{0}$ by
telescoping and using the fact that $\rho_{T}=0$. Due to super-additivity
in losses, we get that $\mu(-\rho_{0})\ge\sum_{t=1}^{T}\mu(\rho_{t}-\rho_{t-1})$
almost surely when the state is bad. Also, because there is \emph{strictly}
gradual bad news, $\mathbb{E}[\sum_{t=1}^{T}\mu(\rho_{t}-\rho_{t-1})\mid\rho_{T}=0]<\mu(-\rho_{0})$.

In the more preferred state, he gets $\mu(1-\rho_{0})$ with one-shot
resolution. With gradual bad news, one-shot good news, let $\hat{T}\le T$
be the first period where $\rho_{\hat{T}}>\rho_{\hat{T}-1}$. His
news utility is $\left[\sum_{t=1}^{\hat{T}-1}\mu(\rho_{t}-\rho_{t-1})\right]+\mu(1-\rho_{\hat{T}-1})$
where each $\rho_{t}-\rho_{t-1}\le0$ for $1\le t\le\hat{T}-1$. Again
by super-additivity in losses, $\sum_{t=1}^{\hat{T}-1}\mu(\rho_{t}-\rho_{t-1})\le\mu(\rho_{\hat{T}-1}-\rho_{0})$.
By sub-additivity in gains, $\mu(1-\rho_{\hat{T}-1})\le\mu(\rho_{0}-\rho_{\hat{T}-1})+\mu(1-\rho_{0})\le-\mu(\rho_{\hat{T}-1}-\rho_{0})+\mu(1-\rho_{0})$,
where the weak inequality follows since $\lambda\ge1.$ Putting these
pieces together, 
\begin{align*}
\left[\sum_{t=1}^{\hat{T}-1}\mu(\rho_{t}-\rho_{t-1})\right]+\mu(1-\rho_{\hat{T}-1}) & \le\mu(\rho_{\hat{T}-1}-\rho_{0})-\mu(\rho_{\hat{T}-1}-\rho_{0})+\mu(1-\rho_{0})=\mu(1-\rho_{0}).
\end{align*}
Therefore, strictly gradual bad news, one-shot good news gives strictly
lower utility than one-shot resolution in expectation, and almost
surely weakly lower utility ex-post.
\end{proof}

\subsection{Proof of Proposition \ref{prop:gradual_good_preferred}}
\begin{proof}
In the preferred state, the agent gets $\mu(1-\rho_{0})$ with one-shot
resolution, but $\sum_{t=1}^{T}\mu(\rho_{t}-\rho_{t-1})$ with gradual
good news, one-shot bad news. For each $t,$ $\rho_{t}-\rho_{t-1}\ge0,$
and furthermore $\sum_{t=1}^{T}\rho_{t}-\rho_{t-1}=1-\rho_{0}$ by
telescoping and using the fact that $\rho_{T}=1$. Due to sub-additivity
in gains, we get that $\sum_{t=1}^{T}\mu(\rho_{t}-\rho_{t-1})\ge\mu(1-\rho_{0})$
when the state is good. Also, because there is \emph{strictly} gradual
good news, $\mathbb{E}[\sum_{t=1}^{T}\mu(\rho_{t}-\rho_{t-1})\mid\rho_{T}=1]>\mu(1-\rho_{0})$.

In the less preferred state, he gets $\mu(-\rho_{0})$ with one-shot
resolution. With gradual good news, one-shot bad news, let $\hat{T}\le T$
be the first period where the $X_{\hat{T}}=0$. His news utility is
$\left[\sum_{t=1}^{\hat{T}-1}\mu(\rho_{t}-\rho_{t-1})\right]+\mu(-\rho_{\hat{T}-1})$
where each $\rho_{t}-\rho_{t-1}\ge0$ for $1\le t\le\hat{T}-1$. Again
by sub-additivity in gains, $\sum_{t=1}^{\hat{T}-1}\mu(\rho_{t}-\rho_{t-1})\ge\mu(\rho_{\hat{T}-1}-\rho_{0})$.
By super-additivity in losses, $\mu(-\rho_{\hat{T}-1})\ge\mu(-(\rho_{\hat{T}-1}-\rho_{0}))+\mu(-\rho_{0})=-\mu(\rho_{\hat{T}-1}-\rho_{0})+\mu(-\rho_{0})$,
where we used the symmetry of $\mu$ around 0 in the last equality.
Putting these pieces together, 
\begin{align*}
\left[\sum_{t=1}^{\hat{T}-1}\mu(\rho_{t}-\rho_{t-1})\right]+\mu(-\rho_{\hat{T}-1}) & \ge\mu(\rho_{\hat{T}-1}-\rho_{0})-\mu(\rho_{\hat{T}-1}-\rho_{0})+\mu(-\rho_{0})=\mu(-\rho_{0}).
\end{align*}
Therefore, strictly gradual good news, one-shot bad news provides
strictly higher utility than one-shot resolution in expectation, and
almost surely weakly higher utility ex-post.
\end{proof}

\subsection{Proof of Corollary \ref{cor:gradual_good_loss_averse}}
\begin{proof}
This follows from Proposition \ref{prop:gradual_good_preferred} by
continuity.
\end{proof}

\subsection{Proof of Corollary \ref{cor:diversity}}
\begin{proof}
When $v(c_{A})=0$ and $v(c_{B})=1$, gradual information falls in
the class of strictly gradual bad news, one-shot good news. So, the
first claim follows from Proposition \ref{prop:opposite}. When $v(c_{A})=1$
and $v(c_{B})=0$, gradual information falls in the class of strictly
gradual good news, one-shot bad news. So, the second claim follows
from Corollary \ref{cor:gradual_good_loss_averse}.
\end{proof}

\subsection{Proof of Proposition \ref{prop:no_diversity}}
\begin{proof}
(1) Suppose $\mu$ is two-part linear with $\mu(x)=x$ for $x\ge0,$
$\mu(x)=\lambda x$ for $x<0,$ where $\lambda\ge0$. Suppose $v(c_{A})=1$,
$v(c_{B})=0.$ In each period, $\mathbb{E}[\mu(\pi_{t}-\pi_{t-1})]=\mathbb{E}[(\pi_{t}-\pi_{t-1})^{+}-\lambda(\pi_{t}-\pi_{t-1})^{-}]$.
By the martingale property, $\mathbb{E}[(\pi_{t}-\pi_{t-1})^{+}]=\mathbb{E}[(\pi_{t}-\pi_{t-1})^{-}]$,
so $\mathbb{E}[\mu(\pi_{t}-\pi_{t-1})]=\frac{1}{2}(1-\lambda)\mathbb{E}[|\pi_{t}-\pi_{t-1}|]$.
This shows total expected news utility is $\mathbb{E}[\sum_{t=1}^{T}\mu(\pi_{t}-\pi_{t-1})]=\frac{1}{2}(1-\lambda)\mathbb{E}[\sum_{t=1}^{T}|\pi_{t}-\pi_{t-1}|]$.
Note that $\mathbb{E}[\sum_{t=1}^{T}|\pi_{t}-\pi_{t-1}|]$ is strictly
larger for gradual information than for one-shot resolution. If $\lambda>1,$
the agent strictly prefers one-shot resolution. If $0\le\lambda<1,$
the agent strictly prefers gradual information. If $\lambda=1,$ the
agent is indifferent.

Now suppose $v(c_{A})=0$, $v(c_{B})=1.$ By the same arguments, total
expected news utility is $\mathbb{E}[\sum_{t=1}^{T}\mu(\rho_{t}-\rho_{t-1})]=\frac{1}{2}(1-\lambda)\mathbb{E}[\sum_{t=1}^{T}|\rho_{t}-\rho_{t-1}|]$.
Note that $\mathbb{E}[\sum_{t=1}^{T}|\rho_{t}-\rho_{t-1}|]$ is strictly
larger for gradual information than for one-shot resolution. So again,
if $\lambda>1,$ the agent strictly prefers one-shot resolution. If
$0\le\lambda<1,$ the agent strictly prefers gradual information.
If $\lambda=1,$ the agent is indifferent.

(2) If $u$ is linear, then the agent is indifferent between gradual
information and one-shot resolution regardless of the sign of $v(c_{A})-v(c_{B})$.
If $u$ is strictly concave, then for $1\le t\le T-1$, $\mathbb{E}[u(\rho_{t})]<u(\rho_{0})$
by combining the martingale property and Jensen's inequality. So the
agent strictly prefer to keep his prior beliefs until the last period
and will therefore choose one-shot resolution, regardless of the sign
of $v(c_{A})-v(c_{B})$.

(3) \citet*{ely2015suspense} mention a ``state-dependent'' specification
of their suspense and surprise utility functions. With two states,
A and B, their specification uses weights $\alpha_{A},\alpha_{B}>0$
to differentially re-scale belief-based utilities for movements in
the two different directions. Specifically, their re-scaled suspense
utility is 
\[
\sum_{t=0}^{T-1}u\left(\mathbb{E}_{t}\left[\alpha_{A}\cdot(\pi_{t+1}-\pi_{t})^{2}+\alpha_{B}\cdot((1-\pi_{t+1})-(1-\pi_{t}))^{2}\right]\right)
\]
and their re-scaled surprise utility is 
\[
\mathbb{E}\left[\sum_{t=1}^{T}u\left(\alpha_{A}\cdot(\pi_{t+1}-\pi_{t})^{2}+\alpha_{B}\cdot((1-\pi_{t+1})-(1-\pi_{t}))^{2}\right)\right].
\]
We may consider agents with opposite preferences over states A and
B as agents with different pairs of scaling weights $(\alpha_{A},\alpha_{B}).$
Specifically, say there are $\alpha^{\text{High}}>\alpha^{\text{Low}}>0$.
For an agent preferring A, $\alpha_{A}=\alpha^{\text{High}},\alpha_{B}=\alpha^{\text{Low}}$.
For an agent preferring B, $\alpha_{A}=\alpha^{\text{Low}},\alpha_{B}=\alpha^{\text{High}}$.
But note that we always have $\pi_{t+1}-\pi_{t}=-[(1-\pi_{t+1})-(1-\pi_{t})]$,
so along every realized path of beliefs, $(\pi_{t+1}-\pi_{t})^{2}=((1-\pi_{t+1})-(1-\pi_{t}))^{2}$.
This means these two agents with the opposite scaling weights actually
have identical objectives and therefore will have the same preference
over gradual information or one-shot resolution.
\end{proof}

\subsection{Proof of Proposition \ref{prop:unique_babbling}}

We begin by giving some additional definition and notation. For $p,\pi\in[0,1],$
let $N_{B}(x;\pi):=\mu(x-\pi)+\mu(-x)$ denote the total amount of
news utility across two periods when the receiver updates his belief
from $\pi$ to $x>\pi$ today and updates it from $x$ to 0 tomorrow.
Similarly, $N_{A}(p;\pi):=\mu(p-\pi)+\mu(1-p)$.

We state some preliminary lemmas about $N_{A}$ and $N_{B}$, whose
proofs appear in Online Appendix \ref{sec:Secondary_proofs}.
\begin{lem}
\label{lem:no_indiff} If $\mu$ is symmetric around 0 and $\mu^{''}(x)<0$
for all $x>0,$ then for any $0<\pi<x<1$ it holds $N_{B}(0;\pi)<N_{B}(x;\pi)$.
\begin{lem}
\label{lem:lemm2} Suppose $\mu$ exhibits diminishing sensitivity
and greater sensitivity to losses. Then, $p\mapsto N_{A}(p;\pi)$
is strictly increasing on $[0,\pi]$ and symmetric on the interval
$[\pi,1]$. For each $p_{1}\in[\pi,1]$, there exists exactly one
point $p_{2}\in[\pi,1]$ so that $N_{A}(p_{1};\pi)=N_{A}(p_{2};\pi).$
For every $p_{L}<\pi$ and $p_{H}\ge\pi,$ $N_{A}(p_{L};\pi)<N_{A}(p_{H};\pi).$
Also, $N_{B}(p;\pi)$ is symmetric on the interval $[0,\pi]$. For
each $p_{1}\in[0,\pi]$, there exists exactly one point $p_{2}\in[0,\pi]$
so that $N_{B}(p_{1};\pi)=N_{B}(p_{2};\pi).$
\end{lem}
\end{lem}
Consider any period $T-2$ history $h_{T-2}$ in any equilibrium $(M,\sigma^{*},p^{*})$
where $p^{*}(h_{T-2})=\pi\in(0,1).$ Let $P_{A}$ and $P_{B}$ represent
the sets of posterior beliefs induced at the end of $T-1$ with positive
probability, in states A and B. The next lemma gives an exhaustive
enumeration of all possible $P_{A},P_{B}$.
\begin{lem}
\label{lem:enumeration}The sets $P_{A},P_{B}$ belong to one of the
following cases.
\end{lem}
\begin{enumerate}
\item $P_{A}=P_{B}=\{\pi\}$
\item $P_{A}=\{1\},$ $P_{B}=\{0\}$
\item $P_{A}=\{p_{1}\}$ for some $p_{1}\in(\pi,1)$ and $P_{B}=\{0,p_{1}\}$
\item $P_{A}=\{\pi,1\}$ and $P_{B}=\{0,\pi\}$
\item $P_{A}=\{p_{1},p_{2}\}$ for some $p_{1}\in(\pi,\frac{1+\pi}{2}),$
$p_{2}=1-p_{1}+\pi$, $P_{B}=\{0,p_{1},p_{2}\}$.
\end{enumerate}
We now give the proof of Proposition \ref{prop:unique_babbling}.
\begin{proof}
Consider any period $T-2$ history $h^{T-2}$ with $p^{*}(h^{T-2})\in(0,1).$
By Lemma \ref{lem:no_indiff}, $N_{B}(p;p^{*}(h^{T-2}))>N_{B}(0;p^{*}(h^{T-2}))$
for all $p\in(p^{*}(h^{T-2}),1]$. Therefore, cases 3 and 5 are ruled
out from the conclusion of Lemma \ref{lem:enumeration}. This shows
that after having reached history $h^{T-2}$, the receiver will get
total news utility of $\mu(1-p^{*}(h^{T-2}))$ in the good state and
$\mu(-p^{*}(h^{T-2}))$ in the bad state. This conclusion applies
to all period $T-2$ histories (including those with equilibrium beliefs
0 or 1). So, the sender gets the same utility as if the state is perfectly
revealed in period $T-1$ rather than $T$, and the equilibrium up
to period $T-1$ form an equilibrium of the cheap talk game with horizon
$T-1.$ By backwards induction, we see that along the equilibrium
path, whenever the receiver's belief updates, it is updated to the
dogmatic belief in $\theta$.
\end{proof}

\subsection{Proof of Proposition \ref{prop:gradual_family}}
\begin{proof}
Let $J$ intermediate beliefs satisfying the hypotheses be given.
We construct a gradual good news equilibrium where $p_{t}=q^{(t)}$
for $1\le t\le J$, and $p_{t}=q^{(J)}$ for $J+1\le t\le T-1.$

Let $M=\{a,b\}$ and consider the following strategy profile. In period
$t\le J$ where the public history so far $h^{t-1}$ does not contain
any $b$, let $\sigma(h^{t-1};A)(a)=1,$ $\sigma(h^{t-1};B)(a)=x$
where $x\in(0,1)$ satisfies $\frac{p_{t-1}}{p_{t-1}+(1-p_{t-1})x}=p_{t}$.
But if public history contains at least one $b,$ then $\sigma(h^{t-1};A)(b)=1$
and $\sigma(h^{t-1};B)(b)=1$. Finally, if the period is $t>J$, then
$\sigma(h^{t-1};A)(b)=1$ and $\sigma(h^{t-1};B)(b)=1$. In terms
of beliefs, suppose $h^{t}$ has $t\le J$ and every message so far
has been $a.$ Such histories are on-path and get assigned the Bayesian
posterior belief. If $h^{t}$ has $t\le J$ and contains at least
one $b$, then it gets assigned belief 0. Finally, if $h^{t}$ has
$t>J$, then $h^{t}$ gets assigned the same belief as the subhistory
constructed from its first $J$ elements. It is easy to verify that
these beliefs are derived from Bayes' rule whenever possible.

We verify that the sender has no incentive to deviate. Consider period
$t\le J$ with history $h^{t-1}$ that does not contain any $b.$
The receiver's current belief is $p_{t-1}$ by construction.

In state $B$, we first calculate the sender's equilibrium payoff
after sending $a.$ The receiver will get some $I$ periods of good
news before the bad state is revealed, either by the sender or by
nature in period $T.$ That is, the equilibrium news utility with
$I$ periods of good news is given by $\sum_{i=1}^{I}\mu(p_{t-1+i}-p_{t-2+i})+\mu(-p_{t-1+I}).$
Since $p_{t-1+I}\in P^{*}(p_{t-2+I})$, we have $N_{B}(p_{t-1+I};p_{t-2+I})=N_{B}(0;p_{t-2+I}),$
that is to say $\mu(p_{t-1+I}-p_{t-2+I})+\mu(-p_{t-1+I})=\mu(-p_{t-2+I}).$
We may therefore rewrite the receiver's total news utility as $\sum_{i=1}^{I-1}\mu(p_{t-1+i}-p_{t-2+i})+\mu(-p_{t-2+I})$.
But by repeating this argument, we conclude that the receiver's total
news utility is just $\mu(-p_{t-1})$. Since this result holds regardless
of $I$'s realization, the sender's expected total utility from sending
$g$ today is $\mu(-p_{t-1})$, which is the same as the news utility
from sending $b$ today. Thus, sender is indifferent between $a$
and $b$ and has no profitable deviation.

In state $A$, the sender gets at least $\mu(1-p_{t-1})$ from following
the equilibrium strategy. This is because the receiver's total news
utility in the good state along the equilibrium path is given by $\sum_{i=1}^{J-(t-1)}\mu(p_{t-1+i}-p_{t-2+i})+\mu(1-p_{t-1+I})$.
By sub-additivity in gains, this sum is strictly larger than $\mu(1-p_{t-1}).$
If the sender deviates to sending $b$ today, then the receiver updates
belief to 0 today and belief remains there until the exogenous revelation,
when belief updates to 1. So this deviation gives the total news utility
$\mu(-p_{t-1})+\mu(1)$. We have 
\begin{align*}
\mu(1) & <\mu(1-p_{t-1})+\mu(p_{t-1})\\
 & \le\mu(1-p_{t-1})-\mu(-p_{t-1}),
\end{align*}
where the first inequality comes from sub-additivity in gains, and
the second from weak loss aversion. This shows $\mu(-p_{t-1})+\mu(1)<\mu(1-p_{t-1})$,
so the deviation is strictly worse than sending the equilibrium message.

Finally, at a history containing at least one $b$ or a history with
length $K$ or longer, the receiver's belief is the same at all continuation
histories. So the sender has no deviation incentives since no deviations
affect future beliefs.

For the other direction, suppose by way of contradiction there exists
a gradual good news equilibrium with the $J$ intermediate beliefs
$q^{(1)}<...<q^{(J)}$. For a given $1\le j\le J,$ find the smallest
$t$ such that $p_{t}=q^{(k-1)}$ and $p_{t+1}=q^{(k)}.$ At every
on-path history $h^{t}\in H^{t}$ with $p^{*}(h^{t})=p_{t}$, we must
have $\sigma^{*}(h^{t};B)$ inducing both 0 and $q^{(j)}$ with strictly
positive probability. Since we are in equilibrium, we must have $\mu(-q^{(j-1)})$
being equal to $\mu(q^{(j)}-q^{(j-1)})$ plus the continuation payoff.
If $j=J$, then this continuation payoff is $\mu(-q^{(j)})$ as the
only other period of belief movement is in period $T$ when the receiver
learns the state is bad. If $j<J,$ then find the smallest $\bar{t}$
so that $p_{\bar{t}+1}=q^{(j+1)}$. At any on-path $h^{\bar{t}}\in H^{\bar{t}}$
which is a continuation of $h^{t},$ we have $p^{*}(h^{\bar{t}})=q^{(j)}$
and the receiver has not experienced any news utility in periods $t+2,...,\bar{t}$.
Also, $\sigma^{*}(h^{\bar{t}};B)$ assigns positive probability to
inducing posterior belief 0, so the continuation payoff in question
must be $\mu(-q^{(j)}).$ So we have shown that $\mu(-q^{(j-1)})=\mu(q^{(j)}-q^{(j-1)})+\mu(-q^{(j)}),$
that is $N_{B}(q^{(j)};q^{(j-1)})=N_{B}(0;q^{(j-1)})$.
\end{proof}

\subsection{Proof of Corollary \ref{cor:P_star_set_quadratic}}
\begin{proof}
We apply Proposition \ref{prop:gradual_family} to the case of quadratic
news utility. Recall the relevant indifference equation in the good
state.
\begin{equation}
\mu(-q_{t})=\mu(q_{t+1}-q_{t})+\mu(-q_{t+1}).\label{eq:indiff}
\end{equation}
Plugging in the quadratic specification and algebraic transformations
lead to
\[
0=(\alpha_{p}-\alpha_{n})(q_{t+1}-q_{t})-\beta_{p}(q_{t+1}-q_{t})+\beta_{n}(q_{t+1}-q_{t})(q_{t+1}+q_{t})
\]
 Define $r=q_{t+1}-q_{t}$. Then this relation can be written as
\[
(\beta_{p}-\beta_{n})r^{2}+(\alpha_{n}-\alpha_{p}-2\beta_{n}q_{t})r=0,
\]
i.e. $r$ is a zero of a second order polynomial. For $P^{*}$ to
be non-empty we need this root $r$ to be in $(0,1-q_{t})$. In particular
the peak/trough $\bar{r}$ of the parabola defined by the second order
polynomial should satisfy $\bar{r}\in(0,\frac{1-q_{t}}{2})$. Given
that $\bar{r}=\frac{2\beta_{n}q_{t}-(\alpha_{n}-\alpha_{p})}{2(\beta_{p}-\beta_{n})}$
for the case that $\beta_{p}\neq\beta_{n}$, we get the equivalent
condition on the primitives $0<\frac{2\beta_{n}q_{t}-(\alpha_{n}-\alpha_{p})}{2(\beta_{p}-\beta_{n})}<\frac{1-q_{t}}{2}.$
The root $r$ itself is given by $r=\frac{2\beta_{n}q_{t}-(\alpha_{n}-\alpha_{p})}{\beta_{p}-\beta_{n}},$
which leads to the recursion 
\begin{equation}
q_{t+1}=q_{t}\frac{\beta_{p}+\beta_{n}}{\beta_{p}-\beta_{n}}-\frac{\alpha_{n}-\alpha_{p}}{\beta_{p}-\beta_{n}}.\label{eq:recursion}
\end{equation}
This leads to the formula for $P^{*}(\pi)$ in part 1).

\textbf{Case 1:} When $\beta_{p}<\beta_{n}$ the coefficient in front
of $q_{t}$ is negative so that the recursion in Equation (\ref{eq:recursion})
leads to 
\[
q_{t+1}-q_{t}=q_{t}\frac{2\beta_{n}}{\beta_{p}-\beta_{n}}-\frac{\alpha_{n}-\alpha_{p}}{\beta_{p}-\beta_{n}}<0.
\]
This also shows that for the case that $\beta_{p}<\beta_{n}$, a GGN
equilibrium with 1 or more intermediate beliefs only exists when the
prior is low enough: namely $\pi_{0}<\frac{\alpha_{n}-\alpha_{p}}{2\beta_{n}}=:q^{*}.$

\textbf{Case 2:} When $\beta_{p}>\beta_{n}$ the slope in Equation
(\ref{eq:recursion}) is above $1$ so that for all priors $\pi_{0}$
large enough we get an increasing sequence $q_{t}$ which satisfies
Equation (\ref{eq:indiff}). It is also easy to see from Equation
(\ref{eq:recursion}) that
\[
(q_{t+2}-q_{t+1})-(q_{t+1}-q_{t})=\left(\frac{\beta_{p}+\beta_{n}}{\beta_{p}-\beta_{n}}-1\right)>0,
\]
proving the statement in the text after the corollary.

That an equilibrium can exist where partial good news are released
for more than two periods, is shown by the example in the main text
following the statement of the Corollary (see Figure \ref{fig:GGN_example}).
\end{proof}

\subsection{Proof of Proposition \ref{prop:ggn_increasing_rate}}
\begin{proof}
Since $N_{B}(p;\pi)-N_{B}(0;\pi)=0$ for $p=\pi$ and $\frac{\partial}{\partial p}N_{B}(p;\pi)|_{p=\pi}>0,$
$N_{B}(p;\pi)-N_{B}(0;\pi)$ starts off positive for $p$ slightly
above $\pi.$ Given that $|P^{*}(\pi)|\le1,$ if we find some $p^{'}>\pi$
with $N_{B}(p^{'};\pi)-N_{B}(0;\pi)>0,$ then any solution to $N_{B}(p;\pi)-N_{B}(0;\pi)=0$
in $(\pi,0)$ must lie to the right of $p^{'}.$

If $q^{(j)},q^{(j+1)}$ are intermediate beliefs in a GGN equilibrium,
then by Proposition \ref{prop:gradual_family}, $q^{(j)}\in P^{*}(q^{(j-1)})$
and $q^{(j+1)}\in P^{*}(q^{(j)})$. Let $p^{'}=q^{(j)}+(q^{(j)}-q^{(j-1)})$.
Then, 
\begin{align*}
N_{B}(p^{'};q^{(j)})-N_{B}(0;q^{(j)}) & =\mu(p^{'}-q^{(j)})+\mu(-p^{'})-\mu(-q^{(j)})\\
 & =\mu(q^{(j)}-q^{(j-1)})+\mu(-q^{(j)}-(q^{(j)}-q^{(j-1)}))-\mu(-q^{(j)})\\
 & >\mu(q^{(j)}-q^{(j-1)})+\mu(-q^{(j-1)}-(q^{(j)}-q^{(j-1)}))-\mu(-q^{(j-1)}),
\end{align*}
where the last inequality comes from diminishing sensitivity. But,
the final expression is $N_{B}(q^{(j)};q^{(j-1)})-N_{B}(0;q^{(j-1)})$,
which is 0 since $q^{(j)}\in P^{*}(q^{(j-1)})$. This shows we must
have $q^{(j+1)}-q^{(j)}>q^{(j)}-q^{(j-1)}.$
\end{proof}

\section{\label{sec:Optimal-Information-Structure}More Results on Preference
over Information Structures}

In this section, we consider an agent who commits to an information
structure at time 0. We find a sufficient condition on the degrees
of loss aversion and diminishing sensitivity for one-shot resolution
to not be optimal across all information structures. When there are
two periods, we find a sufficient condition for all optimal information
structures to satisfy gradual good news, one-shot bad news. Finally,
for two periods and the quadratic news-utility function, we solve
for the optimal information structure in closed form.

\subsection{General Notation for the Information Structure}

In period 0, the agent chooses an information structure $(M,\sigma)$
to learn about the state $\theta$ over time. An information structure
consists of a finite message space $M$ and a family of state-contingent
message distributions $\sigma=(\sigma_{t})_{t=1}^{T-1}$, where $\sigma_{t}(\cdot\mid h^{t-1},\theta)\in\Delta(M)$
is a distribution over messages in period $t$ that depends on the
history $h^{t-1}\in H^{t-1}:=(M)^{t-1}$ of messages so far, as well
as the true state $\theta$.

The agent commits to how he will learn about $\theta$ once he chooses
an information structure. He will mechanically receive messages in
periods $1,2,...,T-1$ according to $\sigma$. At the end of period
$t$ for $1\le t\le T-1$, the agent forms the Bayesian posterior
belief $\pi_{t}$ about the state based on the history $h^{t}\in H^{t}$
of $t$ messages. We normalize $v(c_{A})=1,v(c_{B})=0$, so the agent's
news utility in period $1\le t\le T$ is $\mu(\pi_{t}-\pi_{t-1}).$

\subsection{\label{subsec:suboptimal_oneshot}Sub-Optimality of One-Shot Resolution
and Optimality of Gradual Good News, One-Shot Bad News}

First, we give a sufficient condition on the news-utility function
for one-shot resolution to be strictly suboptimal.
\begin{prop}
\label{prop:subopt_one_shot}For any $T$, one-shot resolution is
strictly suboptimal if 
\[
\mu(1-\pi_{0})-\mu(-\pi_{0})+\mu^{'}(0^{+})-\mu^{'}(1-\pi_{0})+\mu(-1)>0.
\]
\end{prop}
We will soon explain that the condition in Proposition \ref{prop:subopt_one_shot}
has the interpretation of ``strong enough diminishing sensitivity
relative to loss aversion.'' The intuition behind the result, assuming
this interpretation of its condition, is that an information structure
that makes the agent either 99\% certain that the state is good or
fully confident that the state is bad improves on one-shot resolution.
Starting with an interim belief of 0.99, the agent will receive another
small piece of good news with high probability in the future, and
that future news will still deliver sizable positive news utility
by diminishing sensitivity. Of course, the agent will also become
disappointed on rare occasions, but the relative weakness of loss
aversion limits the disutility of this event. (It is easy to show
that if $\mu$ is instead two-part linear, then one-shot resolution
is optimal.)

Now we explain why the condition in this result can be thought of
as a race between loss aversion and diminishing sensitivity. The quadratic
news utility provides a clear illustration of this interpretation:
the condition holds if and only if there is enough curvature relative
to the size of the ``kink'' at 0.
\begin{cor}
\label{cor:quadratic}For quadratic news utility, Proposition \ref{prop:subopt_one_shot}'s
condition is equivalent to $\alpha_{n}-\alpha_{p}<\beta_{n}+\beta_{p}$.
\end{cor}
On the left-hand side, $\alpha_{n}-\alpha_{p}=\mu^{'}(0^{-})-\mu^{'}(0^{+})$,
which is the amount of loss aversion in quadratic news utility near
0. On the right-hand side, $\beta_{p}$ and $\beta_{n}$ control the
amounts of curvature in the positive and negative regions, respectively,
and correspond to the extent of diminishing sensitivity.

The interpretation of ``strong enough diminishing sensitivity relative
to loss aversion'' also extends to a general $\mu.$ We have $\mu(1-\pi_{0})-\mu(-\pi_{0})>0$,
so Proposition \ref{prop:subopt_one_shot}'s condition is satisfied
whenever $\mu^{'}(0^{+})-\mu^{'}(1-\pi_{0})+\mu(-1)>0$. We always
have $\mu^{'}(0^{+})-\mu^{'}(1-\pi_{0})>0,$ and it increases when
$\mu$ becomes more concave in the positive region. We have $\mu(-1)<0,$
but it increases when $\mu$ is more convex in the negative region.
So diminishing sensitivity, in the gains or losses domain, increase
the expression $\mu^{'}(0^{+})-\mu^{'}(1-\pi_{0})+\mu(-1)$. On the
other hand, holding fixed $\mu^{'}(0^{+})$ and the curvature $\mu^{''}(x)$
for $x\ne0$, increasing the amount of loss aversion near 0 (i.e.,
$\mu^{'}(0^{-})-\mu^{'}(0^{+}))$ decreases $\mu(-1)$.

 The next result presents a necessary and sufficient condition for
inconclusive bad news to be suboptimal when $T=2.$ We then verify
the condition for quadratic news utility. Let $U_{1}(\pi_{1}\mid\pi_{0})$
be the sum of the agent's expected news utilities in periods 1 and
2, if he updates his belief to $\pi_{1}\in[0,1]$ at the end of period
1. So $U_{1}(\pi_{1}\mid\pi_{0})=\mu(\pi_{1}-\pi_{0})+\pi_{1}\cdot\mu(1-\pi_{1})+(1-\pi_{1})\cdot\mu(-\pi_{1})$.
\begin{prop}
\label{prop:chord}For $T=2$, information structures with $\mathbb{P}_{(M,\sigma)}[\pi_{1}<\pi_{0}\text{ and }\pi_{1}\ne0]>0$
are strictly suboptimal if and only if there exists some $q\ge\pi_{0}$
so that the chord connecting $(0,U_{1}(0\mid\pi_{0}))$ and $(q,U_{1}(q\mid\pi_{0}))$
lies strictly above $U_{1}(p\mid\pi_{0})$ for all $p\in(0,\pi_{0}).$
\end{prop}
\begin{cor}
\label{cor:chord_quadratic}Quadratic news utility satisfies the condition
of Proposition \ref{prop:chord}.
\end{cor}
If $\mu$ satisfies both conditions in Propositions \ref{prop:subopt_one_shot}
and \ref{prop:chord}, then any optimal information structure for
the agent with $T=2$ must feature strictly gradual good news, one-shot
bad news. In particular, we can combine Corollaries \ref{cor:quadratic}
and \ref{cor:chord_quadratic} to infer that this conclusion applies
to quadratic news utility satisfying $\alpha_{n}-\alpha_{p}<\beta_{n}+\beta_{p}$.

\subsection{\label{subsec:Illustrative-Example:-Quadratic}Explicit Solution
with Quadratic News Utility}

We solve the optimal information structure in closed-form when the
agent has a quadratic news-utility function. Suppose the parameters
of $\mu$ satisfy $\alpha_{n}-\alpha_{p}<\beta_{n}+\beta_{p}$ in
a $T=2$ environment. Furthermore, a concavification argument (Proposition
\ref{prop:backwards_induction}) shows that there exists an optimal
information structure with binary messages that induces either belief
0 or belief $p_{H}>\pi_{0}$ in the only period of communication.
We characterizes $p_{H}$ as the root of a cubic polynomial.
\begin{prop}
\label{prop:quadratic_explicit_concav}For $T=2$ and quadratic news
utility satisfying $\alpha_{n}-\alpha_{p}<\beta_{n}+\beta_{p},$ the
optimal partial good news $p_{H}>\pi_{0}$ satisfies 
\[
\pi_{0}(\alpha_{n}-\alpha_{p})-(\beta_{p}+\beta_{n})\pi_{0}^{2}=p_{H}^{2}(\alpha_{n}-\alpha_{p}+\beta_{n}+\beta_{p})-p_{H}^{3}(2\beta_{p}+2\beta_{n}).
\]
Let $c:=\frac{\alpha_{n}-\alpha_{p}}{\beta_{n}+\beta_{p}}$. We have
$\frac{dp_{H}}{dc}>0.$ Also, we have $\frac{dp_{H}}{d\pi_{0}}<0$
when $\pi_{0}<\frac{1}{2}c$, and $\frac{dp_{H}}{d\pi_{0}}>0$ when
$\pi_{0}>\frac{1}{2}c$.
\end{prop}
There is a tension between loss aversion near the reference point
(captured by $\alpha_{n}-\alpha_{p})$ and diminishing sensitivity
(captured by $\beta_{n}+\beta_{p})$ in shaping the optimal information
structure. Fixing the prior belief, the optimal amount of partial
good news is increasing in loss aversion but decreasing in diminishing
sensitivity. To understand these comparative statics, recall that
in the bad state the agent will sometimes experience false hope as
he gets interim good news. The agent chooses between getting (i) a
larger piece of false hope with lower probability, or (ii) a smaller
piece of false hope with higher probability. When there is more loss
aversion near the reference point, belief paths that feature a small
piece of good news followed by a small piece of bad news become much
more costly, so (i) is preferred. When there is more diminishing sensitivity,
the utility gap between the positive components of (i) and (ii) narrows,
so (ii) becomes more favorable.

The optimal partial good news is non-monotonic in the prior belief
when $\mu$ exhibits loss aversion \textemdash{} in that case, $p_{H}$
decreases with the prior when the prior is low, but increases with
the prior when it is high. Figure \ref{fig:comparative_statics_quadratic_commitment}
illustrates. The intuition is that the agent faces competing incentives
in maximizing his news utility conditional on the good state and the
bad. Conditional on state $A,$ the optimal interim good news is $\frac{1}{2}(1+\pi_{0})$,
which exploits diminishing sensitivity by splitting the good news
evenly across two periods. Conditional on state $B$, the distortion
from loss aversion discussed before pushes towards information structures
that send a bigger piece of interim good news (with lower probability).
For $\pi_{0}$ near 0, the agent's expected welfare is essentially
determined by his welfare in the bad state, so the latter incentive
dominates and $p_{H}$ is far above 0.5. As $\pi_{0}$ increases,
the relative weight on the good state's welfare increases, so $p_{H}$
converges to $\frac{1}{2}(1+\pi_{0}).$ In the case of $\alpha_{n}=\alpha_{p}$,
the distortion from loss aversion is absent, so we get $\frac{dp_{H}}{d\pi_{0}}>0$
for any $\pi_{0}\in(0,1)$.

\begin{figure}[H]
\vspace{-10bp}

\begin{centering}
\includegraphics[scale=0.35]{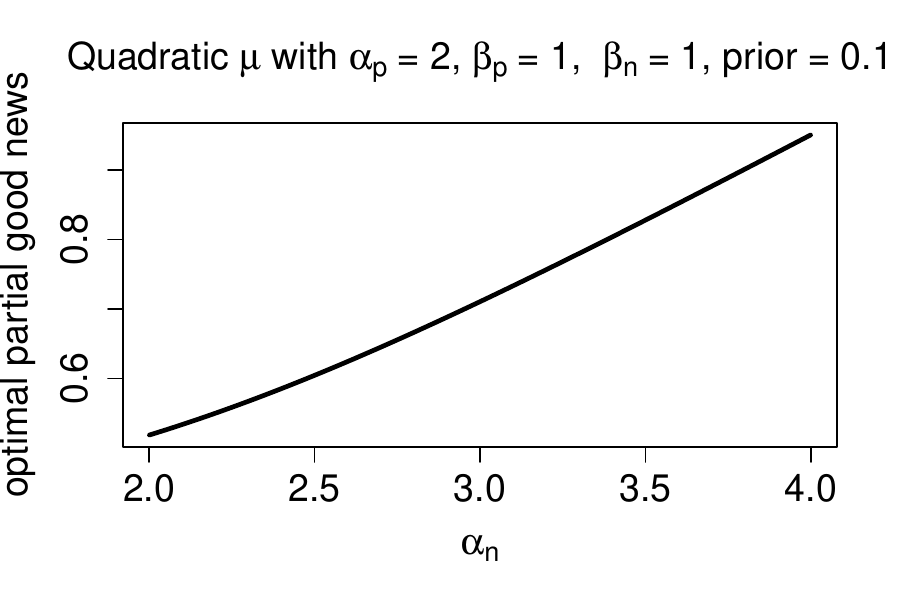} \includegraphics[scale=0.35]{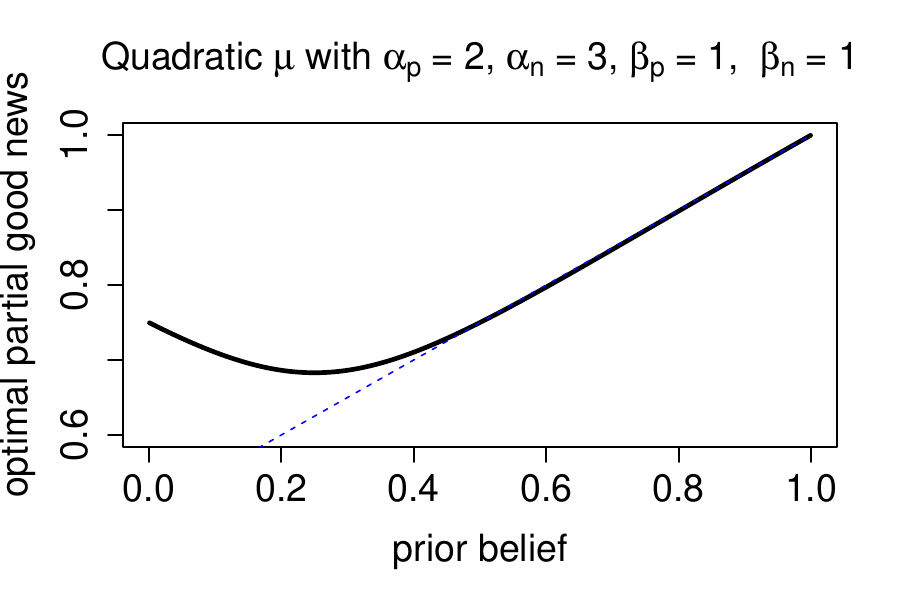}
\par\end{centering}
\vspace{-10bp}

\caption{\label{fig:comparative_statics_quadratic_commitment}\textbf{Left:
}Optimal partial good news with $T=2$, prior $\pi_{0}=0.1,$ quadratic
news utility parameters $\alpha_{p}=2,$ $\beta_{p}=1,$ $\beta_{n}=1$,
as a function of $\alpha_{n}$. The optimal $p_{H}$ monotonically
increases with the amount of loss aversion near 0. \textbf{Right}:
Optimal partial good news with $T=2$, quadratic news utility parameters
$\alpha_{p}=2$, $\alpha_{n}=3,$ $\beta_{p}=1,$ $\beta_{n}=1$,
as a function of the prior belief. The dashed blue line shows $\pi_{0}\protect\mapsto\frac{1}{2}(\pi_{0}+1)$,
the midpoint between the prior and 1. The optimal partial good news
is decreasing in the prior before $\pi_{0}=0.25$, and increasing
afterwards.}
\end{figure}

\subsection{Proofs of Results from Appendix \ref{sec:Optimal-Information-Structure}}

\subsubsection{Proof of Proposition \ref{prop:subopt_one_shot}}
\begin{proof}
Suppose $T=2.$ Consider the following family of information structures,
indexed by $\epsilon>0.$ Let $M=\{m_{L},m_{H}\}.$ Let $\sigma_{1}(A)(m_{H})=1$,
and $\sigma_{1}(B)(m_{L})=x,$ $\sigma_{1}(B)(m_{H})=1-x$ for some
$x\in(0,1)$ so that the posterior belief after observing $m_{H}$
is $(1-\epsilon)$.

For every $\epsilon>0,$ the difference between its expected news
utility and that of one-shot resolution is $W(\epsilon),$ given by
\begin{align*}
 & \pi_{0}\cdot\left[\mu((1-\epsilon)-\pi_{0})+\mu(\epsilon)-\mu(1-\pi_{0})\right]\\
+ & \frac{\epsilon}{1-\epsilon}\pi_{0}\cdot\left[\mu((1-\epsilon)-\pi_{0})+\mu(-(1-\epsilon))-\mu(-\pi_{0})\right].
\end{align*}
$W$ is continuously differentiable away from 0 and $W(0)=0.$ To
show that $W(\epsilon)>0$ for some $\epsilon>0$, it suffices that
$\lim_{\epsilon\to0^{+}}W^{'}(\epsilon)>0.$ Using the continuous
differentiability of $\mu$ except at 0, this limit is 
\[
\mu(1-\pi_{0})-\mu(-\pi_{0})+\mu^{'}(0^{+})-\mu^{'}(1-\pi_{0})+\mu(-1)>0.
\]

If $T>2,$ then note the agent's $T$-period problem starting with
prior $\pi_{0}$ has a value at least as large as the 2-period problem
with the same prior. On the other hand, one-shot resolution brings
the same total expected news utility regardless of $T.$
\end{proof}

\subsubsection{Proof of Corollary \ref{cor:quadratic}}
\begin{proof}
We verify Proposition \ref{prop:subopt_one_shot}'s condition $\mu(1-\pi_{0})-\mu(-\pi_{0})+\mu^{'}(0^{+})-\mu^{'}(1-\pi_{0})+\mu(-1)>0,$
which is equivalent to $\mu^{'}(0^{+})+\mu(1-\pi_{0})-\mu(-\pi_{0})>-\mu(-1)+\mu^{'}(1-\pi_{0}).$
We have that 
\begin{align*}
LHS & =\alpha_{p}+\alpha_{p}(1-\pi_{0})-\beta_{p}(1-\pi_{0})^{2}-[\beta_{n}\pi_{0}^{2}-\alpha_{n}\pi_{0}]\\
RHS & =[-\beta_{n}+\alpha_{n}]+[\alpha_{p}-2\beta_{p}(1-\pi_{0})]
\end{align*}
By algebra, $LHS-RHS=(1-\pi_{0})(\alpha_{p}-\alpha_{n})+(1-\pi_{0}^{2})(\beta_{p}+\beta_{n}).$
Given that $(\alpha_{n}-\alpha_{p})\le(\beta_{p}+\beta_{n})$ and
$1-\pi_{0}^{2}>1-\pi_{0}$ for $0<\pi_{0}<1$, 
\[
LHS-RHS>-(1-\pi_{0}^{2})(\beta_{p}+\beta_{n})+(1-\pi_{0}^{2})(\beta_{p}+\beta_{n})=0.
\]
\end{proof}

\subsubsection{Backwards-Induction with Concavification}

The next few results depend on an iterative concavification procedure.

For $f:\Delta(\Theta)\to\mathbb{R},$ let $\text{cav}f$ be the concavification
of $f$ \textemdash{} that is, the smallest concave function that
dominates $f$ pointwise. Concavification plays a key role in solving
this information design problem, just as in \citet{kg} and \citet{am}.

For $\pi_{T-2},\pi_{T-1}\in[0,1]$ two beliefs about the state, let
$U_{T-1}(\pi_{T-1}\mid\pi_{T-2})$ be the sum of the agent's expected
news utilities in periods $T-1$ and $T$, if he enters period $T-1$
with belief $\pi_{T-2}$ and updates it to $\pi_{T-1}$. More precisely,
\[
U_{T-1}(\pi_{T-1}\mid\pi_{T-2}):=\mu(\pi_{T-1}-\pi_{T-2})+\pi_{T-1}\cdot\mu(1-\pi_{T-1})+(1-\pi_{T-1})\cdot\mu(-\pi_{T-1}),
\]
Note that by the martingale property of beliefs, if the agent holds
belief $\pi_{T-1}$ at the end of period $T-1,$ then state $A$ must
then realize in period $T$ with probability $\pi_{T-1}$.

Let $U_{T-1}^{*}(\pi_{T-2}):=\left(\text{cav}U_{T-1}(\cdot\mid\pi_{T-2})\right)(\pi_{T-2}).$
As we will show in the proof of Proposition \ref{prop:backwards_induction},
$U_{T-1}^{*}(\pi_{T-2})$ is the value function of the agent if he
enters period $T-1$ with belief $\pi_{T-2}.$ Continuing inductively,
using the value function $U_{t+1}^{*}(\cdot)$ for $t\ge1$, we may
define $U_{t}(\pi_{t}\mid\pi_{t-1}):=\mu(\pi_{t}-\pi_{t-1})+U_{t+1}^{*}(\pi_{t}),$
which leads to the period $t$ value function $U_{t}^{*}(x):=\left(\text{cav}U_{t}(\cdot\mid x)\right)(x).$
The maximum expected news utility across all information structures
is $U_{1}^{*}(\pi_{0})$, and the sequence of concavifications give
the optimal information structure.
\begin{prop}
\label{prop:backwards_induction} The maximum expected news utility
across all information structures is $U_{1}^{*}(\pi_{0}).$ There
is an information structure $(M,\sigma)$ with $|M|=2$ attaining
this maximum, with the property that after each on-path public history
$h^{t-1}$ associated with belief $\pi_{t-1},$ the information structure
$\sigma_{t}(\cdot\mid h^{t-1},\theta)$ induces posterior $q^{k}$
at the end of period $t$ with probability $w^{k},$ for some $q^{1},q^{2}\in[0,1],$
$w^{1},w^{2}\ge0,$ $w^{1}+w^{2}=1$, satisfying $\sum_{k=1}^{2}w^{k}q^{k}=\pi_{t-1}$,
and $U_{t}^{*}(\pi_{t-1})=\sum_{k=1}^{2}w^{k}U_{t}(q^{k}\mid\pi_{t-1})$.
\end{prop}
The proof of Proposition \ref{prop:backwards_induction} appears in
Online Appendix \ref{sec:Secondary_proofs}.

\subsubsection{Proof of Proposition \ref{prop:chord}}
\begin{proof}
Suppose the condition in Proposition \ref{prop:chord} holds. So in
particular, it holds for $q=\pi_{0}.$ Consider any information structure
$(M,\sigma)$ and its induced distribution over posterior beliefs
in state $A$, $\eta\in\Delta([0,1]).$ If there exists $0<x<\pi_{0}$
such that $\eta(x)>0,$ then we can ``split posterior $x$ into $0$
and $\pi_{0}$'': that is, we can construct another information structure
$(\tilde{M},\tilde{\sigma})$ with induced distribution $\tilde{\eta}\in\Delta([0,1]),$
so that $\tilde{\eta}(x)=0,$ $\tilde{\eta}(0)=\eta(0)+\eta(x)(1-\frac{x}{\pi_{0}}),$
$\tilde{\eta}(\pi_{0})=\eta(\pi_{0})+\eta(x)\frac{x}{\pi_{0}}$. Information
structure $(\tilde{M},\tilde{\sigma})$ gives strictly higher news
utility than $(M,\sigma)$, since the condition implies $U_{1}(x\mid\pi_{0})<(1-\frac{x}{\pi_{0}})\cdot U_{1}(0\mid\pi_{0})+\frac{x}{\pi_{0}}\cdot U_{1}(\pi_{0}\mid\pi_{0})$.

Conversely, suppose every information structure $(M,\sigma)$ with
induced posterior distribution $\eta$ such that $\eta(x)>0$ for
some $0<x<\pi_{0}$ is strictly suboptimal. Then, there must exist
an optimal information structure, $(\tilde{M},\tilde{\sigma})$ with
posterior distribution $\tilde{\eta},$ so that $\tilde{\eta}$ is
supported on two points: 0 and some $q>\pi_{0}.$ (If no information
at period 1 is optimal, then one-shot resolution at period 1 is also
optimal, which has $\tilde{\eta}(0)>0.$) If there exists some point
$p\in(0,\pi_{0})$ that violates the condition of Proposition \ref{prop:chord}
for this $q$, that is $U_{1}(p\mid\pi_{0})$ is at least as large
as the height of the chord connecting $(0,U_{1}(0\mid\pi_{0}))$ and
$(q,U_{1}(q\mid\pi_{0}))$, then an information structure inducing
the posterior beliefs $p$ and $q$ would strictly dominate the optimal
information structure, which is impossible.
\end{proof}

\subsubsection{Proof of Corollary \ref{cor:chord_quadratic}}

We first state a sufficient condition for the sub-optimality of information
structures with partial bad news with $T=2$. Consider the chord connecting
$(0,U_{1}(0\mid\pi_{0}))$ and $(\pi_{0},U_{1}(\pi_{0}\mid\pi_{0}))$
and let $\ell(x)$ be its height at $x\in[0,\pi_{0}]$. Let $D(x):=\ell(x)-U_{1}(x\mid\pi_{0})$.
\begin{lem}
\label{lem:height_difference_inflection}For this chord to lie strictly
above $U_{1}(p\mid\pi_{0})$ for all $p\in(0,\pi_{0}),$ it suffices
that $D^{'}(0)>0,$ $D^{'}(\pi_{0})<0,$ and $D^{''}(p)=0$ for at
most one $p\in(0,\pi_{0}).$
\end{lem}
The proof of Lemma \ref{lem:height_difference_inflection} appears
in Online Appendix \ref{sec:Secondary_proofs}. Now we verify that
the condition in Lemma \ref{lem:height_difference_inflection} holds
for the quadratic news utility, which in turn verifies the condition
of Proposition \ref{prop:chord} for $q=\pi_{0}$ and shows partial
bad news information structures to be strictly suboptimal.
\begin{proof}
Clearly, $D(p)$ is a third-order polynomial, so $D^{''}(p)$ has
at most one root.

For $p<\pi_{0},$ we have the derivative 
\begin{align*}
\frac{d}{dp}U(p\mid\pi_{0})= & 2\beta_{n}(p-\pi_{0})+\alpha_{n}+\alpha_{p}(1-p)-\beta_{p}(1-p)^{2}\\
 & +p(-\alpha_{p}+2\beta_{p}(1-p))-(\beta_{n}p^{2}-\alpha_{n}p)+(1-p)(2\beta_{n}p-\alpha_{n})
\end{align*}
The slope of the chord between $0$ and $\pi_{0}$ is: $\alpha_{p}-\beta_{p}+(2\beta_{p}-\alpha_{p}+\alpha_{n})\pi_{0}-(\beta_{p}+\beta_{n})\pi_{0}^{2}$.
So, after straightforward algebra, $D^{'}(0)=(2(\beta_{p}+\beta_{n})-(\alpha_{p}-\alpha_{n}))\pi_{0}-(\beta_{p}+\beta_{n})\pi_{0}^{2}.$
Applying weak loss aversion with $z=1$, $\alpha_{p}-\alpha_{n}\le\beta_{p}-\beta_{n}.$
This shows 
\begin{align*}
D^{'}(0) & \ge(2(\beta_{p}+\beta_{n})-(\beta_{p}-\beta_{n}))\pi_{0}-(\beta_{p}+\beta_{n})\pi_{0}^{2}\\
 & =(\beta_{p}+\beta_{n})\pi_{0}(1-\pi_{0})+2\beta_{n}\pi_{0}>0
\end{align*}
for $0<\pi_{0}<1.$

We also derive $D^{'}(\pi_{0})=(\alpha_{p}-2\beta_{p}-2\beta_{n}-\alpha_{n})\pi_{0}+(2\beta_{p}+2\beta_{n})\pi_{0}^{2}.$
Note that this is a convex parabola in $\pi_{0}$, with a root at
0. Also, the parabola evaluated at 1 is equal to $\alpha_{p}-\alpha_{n}\le0,$
where the inequality comes from the weak loss aversion with $z=0$.
This implies $D^{'}(\pi_{0})<0$ for $0<\pi_{0}<1.$
\end{proof}

\subsubsection{Proof of Proposition \ref{prop:quadratic_explicit_concav}}
\begin{proof}
We have 
\[
\frac{d}{dp}U(p\mid\pi_{0})=2\alpha_{p}-\alpha_{n}-\beta_{p}+2\beta_{p}\pi_{0}+p(-2\alpha_{p}+2\beta_{p}+2\alpha_{n}+2\beta_{n})+p^{2}(-3\beta_{p}-3\beta_{n})
\]
Further, $p$ times slope of chord is: 
\begin{align*}
U(p\mid\pi_{0})-U(0\mid\pi_{0})= & U(p\mid\pi_{0})-(\beta_{n}\pi_{0}^{2}-\alpha_{n}\pi_{0})\\
= & \pi_{0}(-\alpha_{p}+\alpha_{n})+\pi_{0}^{2}(-\beta_{p}-\beta_{n})+p(2\alpha_{p}-\alpha_{n}-\beta_{p})\\
 & +p^{2}(-\alpha_{p}+\beta_{p}+\alpha_{n}+\beta_{n})+p^{3}(-\beta_{p}-\beta_{n})+p\pi_{0}(2\beta_{p})
\end{align*}
Equating $p\cdot\frac{d}{dp}U(p\mid\pi_{0})=U(p\mid\pi_{0})-U(0\mid\pi_{0})$,
we get 
\[
\pi_{0}(\alpha_{n}-\alpha_{p})-(\beta_{p}+\beta_{n})\pi_{0}^{2}=p^{2}(\alpha_{n}-\alpha_{p}+\beta_{n}+\beta_{p})-p^{3}(2\beta_{p}+2\beta_{n}).
\]
Define $c=\frac{\alpha_{n}-\alpha_{p}}{\beta_{n}+\beta_{p}}$. Note
that $c\in[0,1)$ by assumptions in the statement of the Proposition.

Corollary \ref{cor:chord_quadratic} allows us to define $p(\pi_{0},c)$
as an implicit function through $\pi_{0}c-\pi_{0}^{2}=p^{2}(1+c)-2p^{3}.$

We characterize first the derivative of $p$ w.r.t. $\pi_{0}$.

We check the conditions of the implicit function theorem in our setting:
define the function $f(\pi_{0},p,c)=p^{2}(1+c)-2p^{3}-\pi_{0}c+\pi_{0}^{2}$
with domain $(0,1)^{3}$. We look at the case $c=0$ separately in
the end. We need $\partial_{p}f(\pi_{0},p,c)\neq0$. If this is true,
then we can solve for $p(\pi_{0},c)$ locally and also calculate its
derivative. We note that $\partial_{p}f(\pi_{0},p,c)=2p(1+c)-6p^{2}$.
Hence, $\partial_{p}f(\pi_{0},p,c)$ is zero if $p=\frac{1+c}{3}=:\hat{p}\in[\frac{1}{3},\frac{2}{3}]$.
Now, for a fixed $c$, we see if there is a $\pi_{0}$ that would
give $\hat{p}$. This involves solving for $\pi_{0}$ in quadratic
equation 
\begin{equation}
\pi_{0}^{2}-\pi_{0}c+\frac{1}{27}(1+c)^{3}=0.\label{eq:q}
\end{equation}
The discriminant as a function of $c$ is given as $D(c)=c^{2}-\frac{4}{27}(1+c)^{3}$.
Note that $D'(c)=\frac{2}{9}(2-c)(2c-1)$. In particular, $D$ is
decreasing from $c=0$ to $c=\frac{1}{2}$ and increasing from then
on until $c=1$. We note also that $D(0)<0,D(1)<0$ so that overall
it follows that $D(c)<0$ for all $c\in[0,1]$. In particular, it
holds that Equation (\ref{eq:q}) has no solution. This means that
$\partial_{p}f$ never changes sign in $(0,1)^{3}\cap\{(p,\pi_{0},c):\pi_{0}c-\pi_{0}^{2}=p^{2}(1+c)-2p^{3}\}$.
This also implies that $p(\pi_{0})>\frac{1+c}{3},$ for all $c,\pi_{0}\in(0,1)$.
Recall here that $f$ is a smooth function on its domain. Thus, implicit
function theorem is applicable for all $(\pi_{0},c)\in(0,1)^{2}$.

Totally differentiating, we get: 
\begin{align*}
d\pi_{0}\cdot(\alpha_{n}-\alpha_{p})-(\beta_{p}+\beta_{n})2\pi_{0}\cdot d\pi_{0} & =2p\cdot dp\cdot(\alpha_{n}-\alpha_{p}+\beta_{n}+\beta_{p})-3p^{2}\cdot dp\cdot(2\beta_{p}+2\beta_{n}),
\end{align*}
which can be rearranged to $\frac{dp}{d\pi_{0}}\frac{1}{p}=\frac{c-2\pi_{0}}{2p^{2}(1+c)-6p^{3}}.$
The steps above showed that the denominator of this expression never
changes sign. Given that we know it is negative at $c=0$ and $f$
is continuously differentiable, we conclude that the denominator is
always negative for all $c$ and all $\pi_{0}\in(0,1)$. It follows
that unless $c=0$, $p(\pi_{0})$ is falling until the prior characterized
in the statement of the Proposition and increasing afterwards.

For the case $c=0$ one looks at the implicitly defined function $p(\pi_{0})$
through $2p^{3}(\pi_{0})-p^{2}(\pi_{0})=\pi_{0}^{2}$. Solving with
similar steps as above one finds, it is strictly increasing and it
is strictly above $\frac{1}{3}$ for all $\pi_{0}\in(0,1)$.

Summarizing, the amount of optimal good news can always be solved
explicitly under the condition that $c\in[0,1)$ and it is always
strictly above $\max\{\pi_{0},\frac{1}{3}\}$. It is strictly increasing
in the prior in the absence of loss aversion and otherwise U-shaped.

Next, we focus on the derivative of $p$ w.r.t. $c$ as defined implicitly
through $\pi_{0}c-\pi_{0}^{2}=p^{2}(1+c)-2p^{3}.$ We fix $\pi_{0}\in(0,1)$
and we look at the function $f:(0,1)\times(0,1)\rightarrow\mathbb{R}$
given by $f(c,p)=p^{2}(1+c)-2p^{3}-(\pi_{0}c-\pi_{0}^{2})$. The calculations
above show that $f_{p}\neq0$ for every $(c,p)\in(0,1)^{2}$. Hence,
the implicit function theorem is applicable and we can calculate 
\[
p\frac{dp}{dc}=\frac{\pi_{0}-p^{2}}{2(1+c)-6p}.
\]
Above we showed that the denominator of this expression is strictly
negative. Next we show that $p(c)>\sqrt{\pi_{0}}$ for all $c\in(0,1)$.
Note that the function $p\mapsto p^{2}(1+c)-2p^{3}$ is strictly decreasing
in $p$ for $p>\frac{1+c}{3}$. This can be established by looking
at first order derivatives. We note also that $\pi_{0}(1+c)-2\sqrt{\pi_{0}}\pi_{0}>\pi_{0}c-\pi_{0}^{2}$
is equivalent to $1-2\sqrt{\pi_{0}}+\pi=(1-\sqrt{\pi_{0}})^{2}>0$,
which is true by virtue of $\pi_{0}<1$. This, and the strict monotonicity
of $p\mapsto p^{2}(1+c)-2p^{3}$ for $p>\frac{1+c}{3}$ implies that
$p(c)>\sqrt{\pi_{0}}$ for all $c,\pi_{0}\in(0,1)$. This establishes
the result.

Summarizing, the amount of optimal good news is increasing in loss
aversion as measured by $\alpha_{n}-\alpha_{p}$ and decreasing in
the amount of diminishing sensitivity, as measured by $\beta_{n}+\beta_{p}$.
\end{proof}
\newpage{}
\begin{center}
\textbf{\Large{}Online Appendix}{\Large\par}
\par\end{center}

\renewcommand{\thesection}{OA \arabic{section}} 
\setcounter{section}{0}
\renewcommand{\thecor}{OA.\arabic{cor}}
\renewcommand{\theprop}{OA.\arabic{prop}}
\renewcommand{\thelem}{OA.\arabic{lem}}
\setcounter{cor}{0}
\setcounter{lem}{0}
\setcounter{prop}{0}
\setcounter{figure}{0}
\setcounter{table}{0}

\section{\label{sec:Secondary_proofs}Proof of Auxiliary Results Stated in
the Appendix}

\subsection{Proof of Lemma \ref{lem:no_indiff}}
\begin{proof}
Due to sub-additivity,
\begin{equation}
\mu(p)<\mu(p-\pi)+\mu(\pi).\label{eq:helplem2}
\end{equation}
Note that symmetry implies $\mu(-p)=-\mu(p)$ and that $\mu(-\pi)=-\mu(\pi)$.
Rearranged (\ref{eq:helplem2}) is precisely $N(0;\pi)<N(p;\pi)$.
\end{proof}

\subsection{Proof of Lemma \ref{lem:lemm2}}
\begin{proof}
We have $\frac{\partial N_{A}(p;\pi)}{\partial p}=\mu^{'}(p-\pi)-\mu^{'}(1-p)$.
For $0\le p<\pi$ and under greater sensitivity to losses, $\mu^{'}(p-\pi)\ge\mu^{'}(\pi-p).$
Since $\mu^{''}(x)<0$ for $x>0,$ $\mu^{'}(\pi-p)>\mu^{'}(1-p)$.
This shows $\frac{\partial N_{A}(p;\pi)}{\partial p}>0$ for $p\in[0,\pi).$

The symmetry results follow from simple algebra and do not require
any assumptions.

Note that $\frac{\partial^{2}N_{A}(p;\pi)}{\partial p^{2}}=\mu^{''}(p-\pi)+\mu^{''}(1-p)<0$
for any $p\in[\pi,1],$ due to diminishing sensitivity. Combined with
the required symmetry, this means $\frac{\partial N_{A}(p;\pi)}{\partial p}$
crosses 0 at most once on $[\pi,1],$ so for each $p_{1}\in[\pi,1]$,
we can find at most one $p_{2}$ so that $N_{A}(p_{1};\pi)=N_{A}(p_{2};\pi)$.
In particular, this implies at every intermediate $p_{1}\in(\pi,1),$
we get $N_{A}(p_{1};\pi)>N_{A}(\pi;\pi)$ since we already have $N_{A}(1;\pi)=N_{A}(\pi;\pi).$
This shows $N_{A}(\cdot;\pi)$ is strictly larger on $[\pi,1]$ than
on $[0,\pi).$

A similar argument, using $\mu^{''}(x)>0$ for $x<0$, establishes
that for each $p_{1}\in[0,\pi]$, we can find at most one $p_{2}$
so that $N_{B}(p_{1};\pi)=N_{B}(p_{2};\pi)$.
\end{proof}

\subsection{Proof of Lemma \ref{lem:enumeration}}
\begin{proof}
Suppose $|P_{A}|=1.$

If $P_{A}=\{\pi\}$, then any equilibrium message not inducing $\pi$
must induce 0. By the Bayes' rule, the sender cannot induce belief
0 with positive probability in the bad state, so $P_{B}=\{\pi\}$
as well.

If $P_{A}=\{1\},$ then any equilibrium message not inducing $1$
must induce 0. Furthermore, the sender cannot send equilibrium messages
inducing belief 1 with positive probability in the bad state, else
the equilibrium belief associated with these messages should be strictly
less than 1. Thus $P_{B}=\{0\}$.

If $P_{A}=\{p_{1}\}$ for some $0\le p_{1}<\pi,$ then any equilibrium
message not inducing $p_{1}$ must induce 0. This is a contradiction
since the posterior beliefs do not average out to $\pi.$

This leaves the case of $P_{A}=\{p_{1}\}$ for some $\pi<p_{1}<1.$
Any equilibrium message not inducing $p_{1}$ must induce 0. Furthermore,
the sender must induce the belief $p_{1}$ in the bad state with positive
probability, else we would have $p_{1}=1.$ At the same time, the
sender must also induce belief 0 with positive probability in the
bad state, else we violate Bayes' rule. So $P_{B}=\{0,p_{1}\}$.

Now suppose $|P_{A}|=2.$

In the good state, the sender must be indifferent between two beliefs
$p_{1},p_{2}$ both induced with positive probability. By Lemma \ref{lem:lemm2},
$N_{A}(p;\pi)$ is strictly increasing on $[0,\pi]$ and strictly
higher on $[\pi,1]$ than on $[0,\pi)$, while for each $p_{1}\in[\pi,1]$,
there exists exactly one point $p_{2}\in[\pi,1]$ so that $N_{A}(p_{1};\pi)=N_{A}(p_{2};\pi).$
This means we must have $p_{1}\in[\pi,\frac{1+\pi}{2}]$, $p_{2}=1-p_{1}+\pi$.

If $P_{A}=\{\pi,1\}$, any equilibrium message not inducing $\pi$
or 1 must induce 0. Also, $1\notin P_{B},$ because any message sent
with positive probability in the bad state cannot induce belief 1.
We cannot have $P_{B}=\{0\}$, because then the message inducing belief
$\pi$ actually induces 1. We cannot have $P_{B}=\{\pi\}$ for then
we violate Bayes' rule. This leaves only $P_{B}=\{0,\pi\}$.

If $P_{A}=\{p_{1},p_{2}\}$ for some $p_{1}\in(\pi,\frac{1+\pi}{2}),$
then any equilibrium message not inducing $p_{1}$ or $p_{2}$ must
induce 0. Also, $p_{1},p_{2}\in P_{B},$ else messages inducing these
beliefs give conclusive evidence of the good state. By Bayes' rule,
we must have $P_{B}=\{0,p_{1},p_{2}\}.$

It is impossible that $|P_{A}|\ge3,$ since, by Lemma \ref{lem:lemm2},
$N_{A}(p;\pi)$ is strictly increasing on $[0,\pi]$ and strictly
higher on $[\pi,1]$ than on $[0,\pi)$, while for each $p_{1}\in[\pi,1]$,
there exists exactly one point $p_{2}\in[\pi,1]$ so that $N_{A}(p_{1};\pi)=N_{A}(p_{2};\pi).$
So the sender cannot be indifferent between 3 or more different posterior
beliefs of the receiver in the good state.
\end{proof}

\subsection{Proof of Proposition \ref{prop:backwards_induction}}
\begin{proof}
We first justify by backwards induction that the value function is
indeed given by $U_{t}^{*}(x)=\left(\text{cav}U_{t}(\cdot\mid x)\right)(x),$
for all $x\in[0,1]$ and all $t\le T-1$, and that it is continuous
in $x$.

If the agent enters period $t=T-1$ with the belief $x\in[0,1]$,
the optimization problem over information structures is the following.

\[
[Q_{T-1}]\quad\max_{\eta\in\Delta([0,1]),\mathbb{E}[\eta]=x}\int_{[0,1]}U_{T-1}(p\mid x)d\eta(p).
\]
This is because any information structure $\sigma_{T-1}$ induces
a Bayes plausible distribution of posterior beliefs, $\eta$ with
$\mathbb{E}[\eta]=x$, and conversely every such distribution can
be generated by some information structure, as in \citet{kg}. It
is well-known that the value of problem $Q_{T-1}$ is $\left(\text{cav}U_{T-1}(\cdot\mid x)\right)(x)$,
justifying $U_{T-1}^{*}(x)$ as the value function for any $x\in\Delta(\Theta)$.
The objective in $Q_{T-1}$ is continuous in $p$ (by assumption on
$\mu$) and hence in $\eta$, and furthermore the constraint set $\{\eta\in\Delta([0,1]):\mathbb{E}[\eta]=x\}$
is continuous in $x$. Therefore, $x\mapsto U_{T-1}^{*}(x)$ is continuous
by Berge's Maximum Theorem.

Assume that we have shown that value function is continuous and given
by $U_{t}^{*}(x)$ for all $t\ge S$. If the agent enters period $t=S-1$
with belief $x,$ then the value of the maximization problem must
be:

\[
[Q_{t}]\quad\max_{\eta\in\Delta([0,1]),\mathbb{E}[\eta]=x}\int_{[0,1]}\mu(p-x)+U_{t+1}^{*}(p)d\eta(p)
\]
 using the inductive hypothesis that $U_{t+1}^{*}(p)$ is the period
$t+1$ value function. But $\mu(p-x)+U_{t+1}^{*}(p)=U_{t}(p\mid x)$
by definition, and it is continuous by the inductive hypothesis. So
by the same arguments as in the base case, $U_{S-1}^{*}(x)$ is the
time-$(S-1)$ value function and it is continuous, completing the
inductive step.

In the first period, by Carath\'{e}odory's theorem, there exist weights
$w^{1},w^{2}\ge0$, beliefs $q^{1},q^{2}\in[0,1],$ with $\sum_{k=1}^{2}w^{k}=1$,
$\sum_{k=1}^{2}w^{k}q^{k}=x$, such that $U_{1}^{*}(\pi_{0})=\sum_{k=1}^{2}w^{k}U_{1}(q^{k}\mid\pi_{0})$.
Having now shown $U_{2}^{*}$ is the period-2 value function, there
must exist an optimal information structure where $\sigma_{1}(\cdot\mid\theta)$
induces beliefs $q^{k}$ with probability $w^{k}$. This information
structure induces one of the beliefs $q^{1},q^{2}$ in the second
period. Repeating the same procedure for subsequent periods establishes
the proposition.
\end{proof}

\subsection{Proof of Lemma \ref{lem:height_difference_inflection}}
\begin{proof}
We need $D>0$ in the region $(0,\pi_{0})$. We know that $D(0)=D(\pi_{0})=0$.
Given the conditions in the statement and the twice-differentiability
of $D$ in $(0,\pi_{0})$ it follows that $D''$ changes sign only
once. Moreover, it also follows that $D>0$ in a right-neighborhood
of $x=0$ and a left-neighborhood of $x=\pi_{0}$. Suppose $D$ has
an interior minimum at $x_{0}\in(0,\pi_{0})$. Then it holds $D''(x_{0})\ge0$.

Suppose $D''(x)>0$ for all small $x$. Then it follows $x_{0}\le p$,
where we set $p=\pi_{0}$ if $p$ doesn't exist. Because $D''(x)\ge0$
for all $x\le p$ we have that $D'(x)>0$ for all $x\le p$. In particular
also $D(x)>0$ for all such $x$ due to the Fundamental Theorem of
Calculus. Thus, the interior minimum is positive and so the claim
about $D$ in $(0,\pi)$ is proven in this case.

Suppose instead that $D''(x)<0$ for all $x$ near enough to $0$.
Then it follows that $x_{0}\ge p$. In particular, for all $x>p$
we have $D''(x)>0$. Since the derivative is strictly increasing for
all $x\in(x_{0},\pi_{0})$ and $D'(\pi_{0})<0$ we have that $D'(x)<0$
for all $x\in(x_{0},\pi_{0})$. In particular, from the Fundamental
Theorem of Calculus, $D(\pi_{0})$ is strictly below $D(x_{0})$.
Since $D(\pi_{0})=0$ we have again that $D(x_{0})>0$.

Given the boundary values of $D$ and the signs of the derivatives
at $0,\pi_{0}$ and that any interior minimum of $D$ is strictly
positive, we have covered all cases and so shown that $D>0$ in $(0,\pi_{0})$.
\end{proof}

\section{Further Results}

\subsection{\label{subsec:Preference-for-Dominated}Preference for Dominated
Consumption Lotteries}

So far, we have taken the prior distribution over states $\pi_{0}\in\Delta(\Theta)$
as exogenously given. Fixing an information structure, a news-utility
agent may strictly prefer a dominated distribution over states. This
distinguishes our news-utility preference from other preferences,
such as recursive preferences and \citet*{gul2019random}'s risk consumption
preference.

We now give an example. Suppose $T=2$ and there are two states, $\Theta=\{A,B\}.$
Normalize consumption utility to be $v(c_{A})=1,$ $v(c_{B})=0.$
Let the news utility function be $\mu(z)=\sqrt{z}$ for $z\ge0,$
$\mu(z)=-\lambda\sqrt{-z}$ for $z<0,$ where $\lambda\ge1.$ At time
$t=0,$ the agent holds a prior belief $\pi_{0}$ with $\pi_{0}(A)=p\in[0,1].$
At time $t=1,$ the agent learns the state perfectly, so $\pi_{1}$
is degenerate with probability 1. Consumption takes place at time
$t=2.$ For any $\lambda,$ the agent strictly prefers state $A$
for sure ($\pi_{0}(A)=1$) over state $B$ for sure ($\pi_{0}(A)=0$),
as both environments provide zero news utility. But, the agent may
strictly prefer state $B$ for sure over an interior probability of
the good state, $\pi_{0}(A)=p.$ In fact, this happens when $p+p\sqrt{1-p}-\lambda(1-p)\sqrt{p}<0$,
which says $\lambda>\frac{\sqrt{p}(1+\sqrt{1-p})}{1-p}$. A sufficiently
loss-averse agent may strictly prefer no chance of winning a consumption
lottery than a low chance of winning.

\subsection{\label{subsec:Optimal-Information-Anticipatory}Optimal Information
Structure for Anticipatory Utility}

We show that if the agent has anticipatory utility and gets $A\left(\sum\pi_{t}(\theta)\cdot v(c_{\theta})\right)$
when he ends period $t$ with posterior belief $\pi_{t}\in\Delta(\Theta),$
then with commitment power, there exists an optimal information structure
that only discloses information in period $t=1$.

Consider any information structure $(M,\sigma).$ Find the period
$t^{*}$ with the highest ex-ante anticipatory utility, i.e., $t^{*}\in\underset{1\le t\le T-1}{\arg\max}\mathbb{E}_{(M,\sigma)}\left[A\left(\sum\pi_{t}(\theta)\cdot v(c_{\theta})\right)\right]$.
Consider another information structure that generates the (feasible)
distribution of beliefs $\pi_{t^{*}}$ in period 1, then reveals no
additional information in periods $2,...,T-1.$ This new information
structure gives weakly higher expected anticipatory utility than $(M,\sigma)$
in every period. Therefore there exists an optimal information structure
that only discloses information in $t=1.$

\subsection{\label{subsec:Risk-Consumption-Preferences}Risk Consumption Preferences}

\citet*{gul2019random} study a model of preference over random evolving
lotteries and propose a class of risk consumption preferences. Translated
into our setting, an agent with risk consumption preference values
an information structure $(M,\sigma)$ according to utility function
\[
\mathbb{E}_{(M,\sigma)}\left[\int v(u_{2}(\pi_{t}))d\eta\right].
\]
 Here $u_{2}:\Delta(\Theta)\to\mathbb{R}$ is affine and $v$ is strictly
increasing. The term $v(u_{2}(\pi_{t}))$ is viewed as a function
from the time periods $\{0,1,...,T-1\}$ into the reals and $d\eta$
denotes the Choquet integral with respect to a capacity $\eta$ on
$\{0,1,...,T-1\}.$

To show that our model of mean-based news utility is not nested under
the class of risk consumption preferences, we show that risk consumption
preferences cannot exhibit the preference patterns from Appendix \ref{subsec:Preference-for-Dominated}:
that is, strictly preferring winning a lottery for sure to not winning
it for sure, but also strictly preferring not winning for sure to
winning with some interior probability $p\in(0,1)$ in the $T=2$
setup.

By an abuse of notation, the belief assigning probability $q$ to
state $A$ will simply be denoted $q.$ The first part of the preference
gives $v(u_{2}(1))>v(u_{2}(0))$, since Choquet integral of a constant
function returns the same constant. When the prior winning probability
is $p\in(0,1),$ the Choquet integrand is either $f_{A}:\{0,1\}\to\mathbb{R}$
with $f_{A}(0)=v(u_{2}(p))$ and $f_{A}(1)=v(u_{2}(1))$, or $f_{B}:\{0,1\}\to\mathbb{R}$
with $f_{B}(0)=v(u_{2}(p))$ and $f_{B}(0)=v(u_{2}(0))$. The two
integrands correspond to belief paths where the agent wins or loses
the lottery. Since $v$ is strictly increasing, $u_{2}$ is affine,
and $v(u_{2}(1))>v(u_{2}(0)),$ we have $v(u_{2}(p))>v(u_{2}(0)).$
Thus both $f_{G}$ and $f_{B}$ dominate the constant function $v(u_{2}(0))$
in every period. By monotonicity of the Choquet integral, the agent
must prefer $p$ probability of winning the lottery to no chance of
winning it.
\end{document}